\documentclass[12pt]{colt2019} 

\title[Nearly Optimal Dynamic $k$-Means Clustering for High-Dimensional Data]{Nearly Optimal Dynamic $k$-Means Clustering \\for High-Dimensional Data}
\usepackage{times}
\usepackage{amsmath}
\usepackage{amssymb}
\usepackage{algorithmic}
\usepackage{algorithm}
\usepackage{color}
\usepackage[english]{babel}
\usepackage{graphicx}
\usepackage{wrapfig}
\usepackage{epstopdf}
\usepackage{url}
\usepackage{graphicx}
\usepackage{color}
\usepackage{epstopdf}
\usepackage{scrextend}
\usepackage[T1]{fontenc}
\usepackage{bbm}
\usepackage{comment}
\usepackage{xcolor}
\newcommand{\codecmt}[1]{\textcolor{blue!10!black!80!green!70!}{{\small$\triangleright$~#1}}}

\newtheorem{claim}[theorem]{Claim}
\newtheorem{fact}[theorem]{Fact}

\newcommand{\wh}{\widehat}
\newcommand{\wt}{\widetilde}

\newcommand{\eps}{\epsilon}

\newcommand{\R}{\mathbb{R}}

\newcommand{\norm}[1]{\left\lVert#1\right\rVert}
\renewcommand{\varepsilon}{\epsilon}
\renewcommand{\tilde}{\wt}
\renewcommand{\hat}{\wh}

\DeclareMathOperator*{\E}{{\bf {E}}}
\DeclareMathOperator*{\Var}{{\bf {Var}}}

\DeclareMathOperator{\OPT}{\mathsf{OPT}}

\DeclareMathOperator{\poly}{poly}

\DeclareMathOperator{\dist}{dist}

\DeclareMathOperator{\cost}{cost}

\DeclareMathOperator{\diam}{diam}

\newenvironment{proofof}[1]
{
	\par\noindent{\bfseries\upshape Proof of #1\ }%
}%
{\jmlrQED}

{\jmlrQED}

{\jmlrQED}

\allowdisplaybreaks

\makeatletter
\newcommand*{\RN}[1]{\expandafter\@slowromancap\romannumeral #1@}
\makeatother
\coltauthor{%
 \Name{Wei Hu} \Email{huwei@cs.princeton.edu}\\
 \addr Princeton University
 \AND
 \Name{Zhao Song} \Email{magic.linuxkde@gmail.com}\\
 \addr %
 \AND
 \Name{Lin F. Yang} \Email{lin.yang@princeton.edu}\\
 \addr Princeton University%
 \AND
 \Name{Peilin Zhong} \Email{pz2225@columbia.edu}\\
 \addr Columbia University%
}

\begin{document}

\maketitle

\begin{abstract}%

We consider the $k$-means clustering problem in the dynamic streaming setting, where points from a discrete Euclidean space $\{1, 2, \ldots, \Delta\}^d$ can be dynamically inserted to or deleted from the dataset. For this problem, we provide a one-pass coreset construction algorithm using space $\tilde{O}(k\cdot \poly(d, \log\Delta))$, where $k$ is the target number of centers. To our knowledge, this is the first dynamic geometric data stream algorithm for $k$-means using space polynomial in dimension and nearly optimal (linear) in $k$.

\end{abstract}



\section{Introduction}

Clustering is one of the central problems in unsupervised learning.
The idea is to partition data points into clusters in the hope that points in the same cluster are similar to each other and points in different clusters are dissimilar.
One of the most important approaches to clustering is \emph{$k$-means}, which has been extensively studied for more than 60 years and has a wide range of applications (see e.g. \citep{j10} for a survey).
Given a set of points $Q\subset \R^d$, the $k$-means problem asks for a set of $k$ \emph{centers} $Z \subset \R^d$ such that the sum of squares of distances between data points to their closest centers is minimized, i.e., it tries to solve $\min\limits_{Z \subset \R^d, |Z|=k} \cost(Q, Z)$, where $\cost(Q, Z)$ is a cost function defined as:
\begin{equation*}
\cost(Q, Z) := \sum_{q\in Q} \min_{z\in Z} \dist^2(q, z).
\end{equation*}
Here $\dist(\cdot, \cdot)$ stands for the Euclidean distance.


A major challenge in dealing with massive datasets is that the entire input data can be too large to be stored.
A standard model of study in such settings is the \emph{streaming} model, where data points arrive and are processed one at a time, and only a small amount of useful information (i.e., a \emph{sketch}) about the data is maintained.
See~\citep{muthukrishnan2005data} for an introduction to the streaming model.

In this paper we study the $k$-means problem over \emph{dynamic data streams}~\citep{i04}, where data points from a discrete space $\{1, 2, \ldots,\Delta\}^d$ can be either inserted to or deleted from the dataset.
A standard approach to solving $k$-clustering problems like $k$-means and $k$-median in the streaming setting is to maintain an \emph{$\epsilon$-coreset}, which is a small number of (weighted) points whose cost with respect to any $k$ centers is a $(1+\epsilon)$-approximation to the cost of the entire dataset on the same $k$ centers.
As a consequence, at the end of the stream, we only need to find an approximate $k$-means solution on the coreset, which is automatically an approximate solution on the entire dataset.
Hence our goal is to design an efficient method to maintain an $\epsilon$-coreset over a dynamic data stream using as small space as possible.

\subsection{Our Result}


\begin{theorem}[Main theorem, restatement of Theorem~\ref{thm:main}]\label{thm:main_restate}
	Let $\epsilon\in (0,1/2)$, $k, \Delta \in \mathbb{N}_+$, and $L=\log \Delta$.
	For dynamic data stream consisting of insertions and deletions of points in $[\Delta]^d$, there is an algorithm which uses a single pass over the stream and on termination outputs a weighted set $S$ with a positive weight for each point therein, such that with probability at least $0.9$, $S$ is an $\epsilon$-coreset for $k$-means of size
	$
	O(k\varepsilon^{-2}d^4L^2\log(kdL))
	$. 
	The algorithm uses
	$
	\wt{O}(k) \cdot\poly(d, L, \epsilon^{-1}) 
	$
	bits in the worst case. 
\end{theorem}
To our knowledge, this is the first algorithm for $k$-means in dynamic data streams that uses space \emph{polynomial in data dimension $d$} and \emph{nearly optimal (linear)\footnote{It is easy to see that $k$ points are needed in a coreset -- when there are only $k$ points in the dataset, the optimal $k$-means cost is $0$, so a coreset has to contain all $k$ points.} in the number of clusters $k$}.
Previous algorithms for streaming $k$-means either require space exponential in $d$ or only work for insertion-only streams.\footnote{It is also possible to obtain an $\tilde{O}(k^2\cdot\poly(d))$ space algorithm for dynamic streams by combining the techniques from \citep{chen09} and \citep{bflsy17}.}
See Section~\ref{sec:related} and Appendix~\ref{sec:previous algorithm fail} for detailed discussions of previous results.


Note that for the $k$-means problem, \cite{cemmp15} showed that one can always do a random projection to reduce the dimension to $O(k/\varepsilon^2)$.
Thus, the most interesting setting would be when $d\le O(k/\eps^2)$ and $d \gg\log k$.

\subsection{Our Techniques}


At a high level our algorithm is based on a framework called \emph{sensitivity sampling}, which was proposed by \cite{fl11}.
For a set $Q \subseteq [\Delta]^d$, the \emph{sensitivity} of every point $q\in Q$ is defined as
\[
s(q) := \max_{Z\subset\R^d, |Z|=k} \frac{\dist^2(q, Z)}{\sum_{p\in Q}\dist^2(p, Z)}.
\]
Namely, $s(q)$ represents how ``sensitive'' the cost can be to the removal of point $q$.
A crucial result shown by \cite{fl11,bfl16} is that once we know a good \emph{upper bound} on each point's sensitivity, there is a sampling method to construct an $\epsilon$-coreset.
Specifically, if we know an upper bound $s'(q) \ge s(q)$ for each $q\in Q$, we can sample $q$ with probability $s'(q)/(\sum_{p\in Q}s'(p))$.
Let $R$ be a set of i.i.d. samples from this procedure with $|R|\geq \tilde{\Omega}\left( \sum_{q\in Q}s'(q)/\epsilon^2\right)$, and each sample $q$ is assigned a weight $\frac{\sum_{p\in Q}s'(p)}{|R|s'(q)}$.
Then with high probability $R$ is an $\epsilon$-coreset for $Q$.
Note that if $\sum_{q\in Q} s'(q) = \wt{O}(k \cdot \poly(d))$, then an $\wt O(k \cdot \poly(d))$-size $\epsilon$-coreset can be constructed in this way.
The formal description of this result is given in Theorem~\ref{thm:fl11}.

We give an efficient method to obtain sensitivity upper bounds $s'(\cdot)$ such that: (i) $\sum_{q\in Q} s'(q)$ is small, (ii) we can implement the sensitivity sampling procedure in the dynamic streaming setting. 
Then we are able to construct a coreset according to the previous paragraph.

The key intuition in our sensitivity estimation is the following. Imagine that there is a small region that is very dense, i.e., it contains a lot of points. Then the sensitivity of every point in that region must be low, because that point can be well represented by other points in the same small region.
Therefore, the problem of finding sensitivity upper bound for a point boils down to figuring out the ``right'' region this point belongs to that can be considered ``dense.''
Intuitively the sensitivity of this point depends on the size of this dense region -- the smaller the size, the smaller the sensitivity.

\begin{figure}[t]
	\centering
	\includegraphics[width=0.25\textwidth]{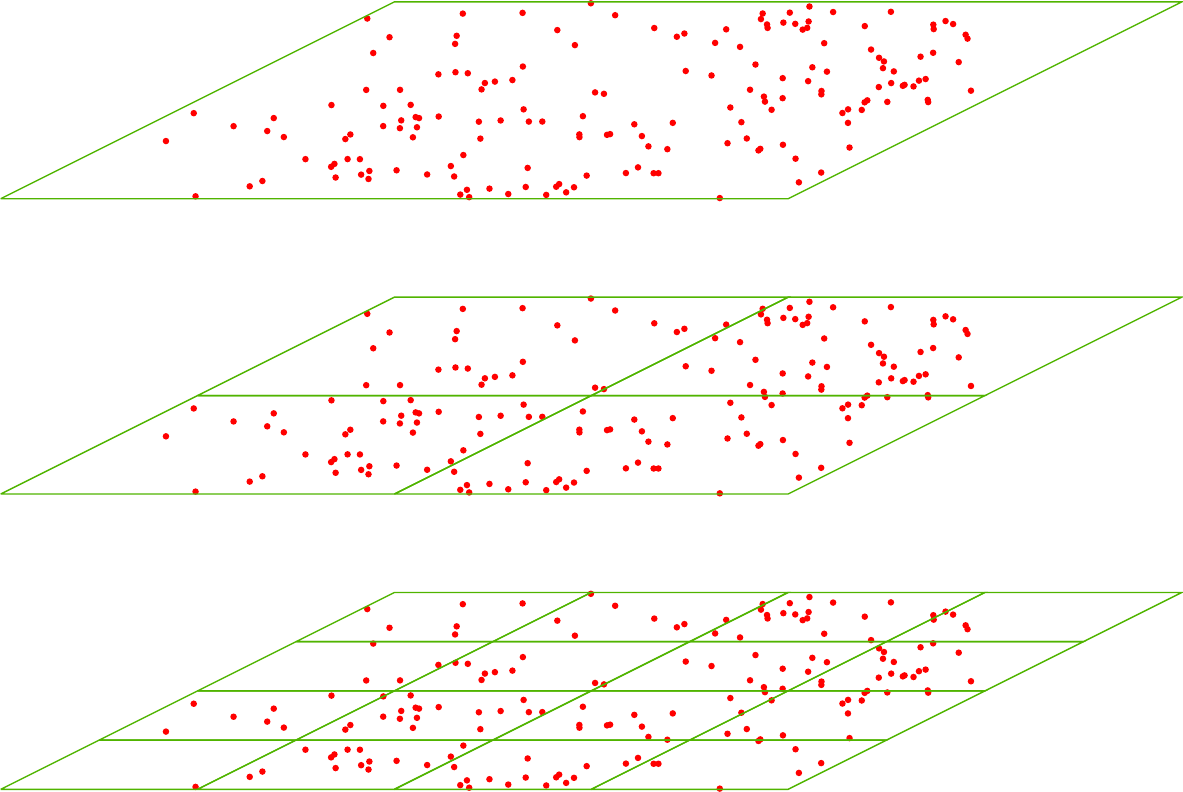}
	\caption{
		\small The grid structure over the point set.
		From top to bottom, three levels of grids are shown.
		Each cell splits into $2^d$ cells in the next level.}
	\vspace{-2em}
	\label{fig:grid}
\end{figure}

We make this intuition formal by using a hierarchical grid structure similar to~\citep{fs05,bflsy17}.
This structure is illustrated in Figure~\ref{fig:grid}.
The top-level (level $0$) grid consists of cells that are $d$-dimensional cubes of side-length $\Delta$, and each cell in level $i-1$ splits into $2^d$ cells in level $i$.
Each cell in level $i$ has side-length $\Delta/2^i$.
For a cell in level $i$, we say that it is \emph{heavy} if it contains at least $T_i = \Theta\left(\frac{d^2}{k}\cdot \frac{\OPT}{(\Delta/2^i)^2}\right)$ points in $Q$, where $\OPT$ is the optimal cost of the $k$-means problem.\footnote{We assume for now that we know $\OPT$. Our actual algorithm uses exponential search to guess the value of $\OPT$.}
Since $T_i > T_{i-1}$, we know that if a cell in level $i$ is heavy, then its parent cell in level $i-1$ is heavy as well. Therefore the set of all heavy cells in all levels form a tree.
Now for a point $p\in Q$, denote by $c_i(p)$ the cell in level $i$ that contains $p$, and then define $j$ to be the smallest level index such that $c_j(p)$ is not heavy; 
then we show an upper bound on the sensitivity $s(p)$ solely based on this index number $j$, namely $s(p) \le s'(p) = \Theta(d^3/T_j)$.
Furthermore, we prove that the sum of our sensitivity upper bounds is small: $\sum_{p\in Q}s'(p) = O\left( kd^3 \log\Delta \right)$, which satisfies our requirement.
To establish these bounds we need the total number of heavy cells to be small, for which we apply a random shift of grid at the beginning, as illustrated in Figure~\ref{fig:random shift}.


To implement the above sensitivity sampling method in the dynamic streaming setting when the dataset is updated by insertions and deletions of points,
the key difficulties are: 1) we do not know the value of $\OPT$, and it changes when the underlying dataset is updated;
2) we need to compute the sensitivity upper bounds and to sample points at the same time using limited space.

Let us first assume $\OPT$ is known and give an algorithm to implement our sensitivity sampling procedure in the dynamic streaming setting.
Our algorithm makes crucial use of the \emph{$k$-set} data structure in \citep{ganguly2005counting} for counting \emph{distinct elements} in a dynamic stream.
The $k$-set data structure ensures that if the number of distinct elements is at most some predetermined parameter, it will return all distinct elements and their frequencies; otherwise it will return \textbf{FAIL}.
We summarize its guarantee in Lemma~\ref{lem:k_set}.
Note that in order to implement sensitivity sampling, we need to know which cells are heavy.
Our algorithm dynamically tracks all heavy cells, using the $k$-set structure as a building block.
Then the sensitivity sampling method has two stages: first sample a level $i$ (with an appropriate probability for each level), and then uniformly sample a point from all points associated with level $i$, i.e., all points $p$ such that $c_i(p)$ is not heavy and  $c_{i-1}(p)$ is heavy. (Note that for all points associated with level $i$, they have the same sensitivity upper bounds, which means uniformly sampling a point from them is enough.)
In order to do uniform sampling, we also maintain for each level $i$ a uniformly random subset of points associated with $i$. Therefore it suffices to choose a point uniformly at random from this subset once $i$ is chosen.

For the issue of not knowing $\OPT$, we run in parallel multiple copies of our sampling algorithm for different guesses of $\OPT$: $1,2,4\ldots, \Delta^d \cdot d\Delta$.
Our sampling algorithm ensures that when the guessed value is less than $\OPT$ but not too far away, the required space is small. For other guesses, the required space might be a lot, but since we have a space budget, our algorithm can return \textbf{FAIL} when the space runs out. Since at least one guess is accurate, at least one copy of the algorithm will succeed and output a small $\epsilon$-coreset.

\begin{figure}[t]
	\centering
	\includegraphics[width=0.7\textwidth]{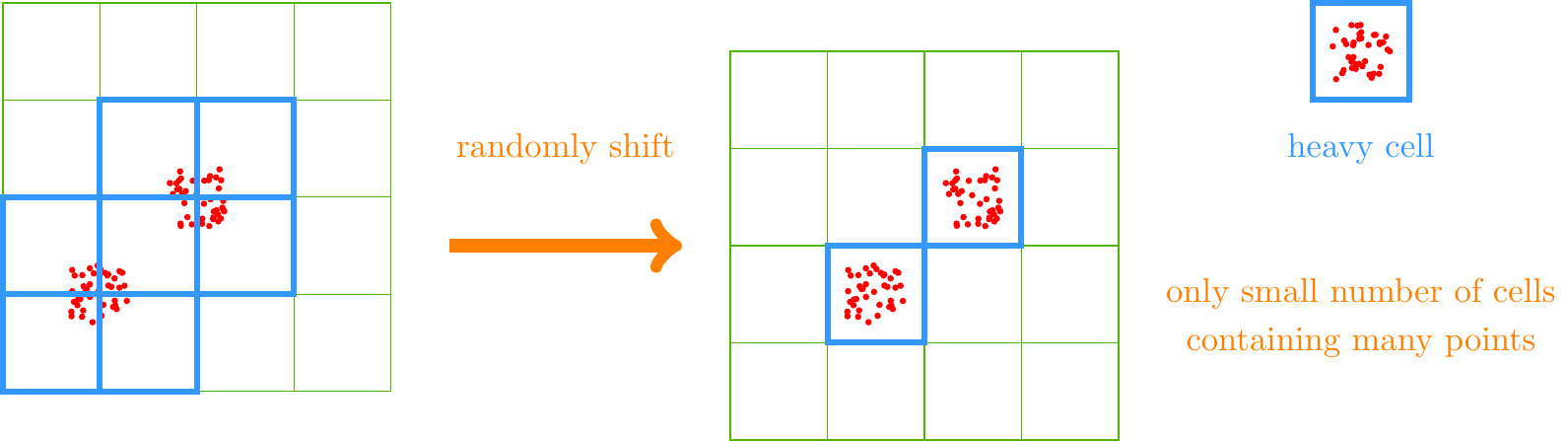}
	\vspace{-0.15in}
	\caption{\small Random shift of grid brings down the number of heavy cells.
		In the left panel, we have a bad alignment of points and grids such that many cells contain lots of points. In the right panel, after the random shift, only two cells contain many points. \label{fig:random shift}}
		\vspace{-0.2in}
\end{figure}

\subsection{Related Work}\label{sec:related}

It is well known that exactly solving $k$-means is NP-hard even for $k=2$ or $d=2$ \citep{adhp09,mnv09}. The most successful algorithm used in practice is Lloyd's algorithm, which is also known as ``the'' $k$-means method \citep{l82}.
Because of the NP-hardness, various attempts were made on approximation algorithms.
\cite{kmnpsw02} proved that a very simple local search heuristic achieves $(9+\epsilon)$-approximation in polynomial time for any fixed $\epsilon >0$. 
When $d$ is a constant \citep{frs16,ckm16} or $k$ is a constant \citep{fms07,kss10,fl11}, $(1+\epsilon)$-approximation can be achieved in polynomial time.

There is a line of work studying $k$-means and $k$-median in insertion-only streams, e.g.,  \citep{bs80, gmmo00, cop03, bdmo03, ahv04, hm04, hk05, chen09, fl11, fs12, amrsls12, bfl16}. 
There also have been a lot of interests in dynamic streaming algorithms for other problems, e.g.
\citep{bjkst02,fkmsz05,b08,kl11,agm12b,agm12,gkp12,gkk12,agm12,bks12,cms13,agm13,m14,bgs15,bhnt15,bs16,acdgw16,adkkp16,bwz16,klmms17,swz17,swz17b}. 
In addition, $k$-means and $k$-median were studied in various different settings, e.g., \citep{ccgg98,ip11,birw16,bcmn14,sw18}.

The most relevant papers are \citep{fs05, fl11, bfl16, bflsy17}.
\cite{fs05} designed an algorithm to maintain an $\epsilon$-coreset of size $k\epsilon^{-O(d)}$ for $k$-means and $k$-median.
\cite{fl11} introduced the sensitivity sampling framework for coreset construction, and their approach was further improved by \cite{bfl16}, but both of them only work for insertion-only streams and do not apply to dynamic streams.
\cite{bflsy17} focused on the $k$-median problem and constructed a coreset of size $O(k \cdot \poly(d, \log\Delta))$ in the dynamic streaming setting, but their technique heavily relies on $k$-median and cannot be extended to $k$-means.
In Appendix~\ref{sec:previous algorithm fail} we explain in detail the limitations of previous approaches.

\section{Preliminaries}\label{sec:preli}


\paragraph{Notation.} For $n\in \mathbb{N}_{+}$, let $[n]:=\{1,2,\ldots,n\}$. 
We define $\wt{O}(f)$ to be $O\left( f\cdot \log^{O(1)}(f) \right)$. 
For any $x\in \mathbb{R}_{\ge 0},\varepsilon\in(0,1),$ we use $(1\pm \varepsilon) \cdot x$ to denote the interval $((1-\varepsilon)\cdot x, (1+\varepsilon)\cdot x)$. For any $x\in\mathbb{R}$ and $a\in\mathbb{R}_{> 0}$, $x\pm a$ denotes the interval $(x-a,x+a)$.

We denote by $\dist(\cdot, \cdot)$ the Euclidean distance in $\mathbb{R}^d$, i.e., for $p, q\in \mathbb{R}^d$,
$
\dist(p, q) := \norm{p-q}_2
$.
For sets $P, Q\subseteq \mathbb{R}^d$ and a point $p\in \mathbb{R}^d$, we define
$
\dist(p, Q) = \dist(Q, p) := \min_{q\in Q}\dist(p,q)$ and $\dist(P, Q):=\min_{p\in P, q\in Q}\dist(p,q).
$
For any two sets $Q,Z\subseteq \mathbb{R}^d$, we define $\cost(Q,Z):=\sum_{q\in Q}\dist^2(q,Z)$.
We define $\diam(Q)$  to be $Q$'s diameter, i.e., $\diam(Q) := \max_{p, q\in Q} \dist(p, q)$.

\paragraph{The dynamic streaming model.}
We consider the \emph{dynamic streaming model}, defined below. 
\begin{definition}[Dynamic streaming model]\label{def:dynamic_stream}
Let $Q\subseteq[\Delta]^d$ initially be an empty set. In the dynamic streaming model, there is a stream of update operations such that the $t^{\text{th}}$ operation has the form $(p_t,\pm)$ which indicates that a point $p_t\in[\Delta]^d$ is inserted to or deleted from the set $Q$, where $+$ denotes insertion and $-$ denotes deletion. 
There is no invalid deletion during the stream.\footnote{At any time during the stream, for any point $p\in[\Delta]^d$, the number of deletions of $p$ so far is always no more than the number of insertions of $p$.}
An algorithm is allowed a single pass over the stream. At the end of the stream, the algorithm stores some information regarding $Q$. The space complexity of an algorithm in this model is defined as the total number of bits used by the algorithm during the stream.  
\end{definition} 

The goal of an algorithm in this model is to store some information which can be used for a certain computation  task, while using as small space as possible. Although optimizing the running time 
is not required in this model, the algorithm in the current paper is actually efficient for each update.

In this paper
we suppose that any two points in $Q$ have different locations\footnote{At the end of the stream, for any point $p\in[\Delta]^d$, the number of insertions of $p$ is at most one more than the number of deletions of $p$. }, i.e., $Q$ is not a multiset. Our algorithm can be easily extended to allow multiple copies of a point by blowing up the total space by an $O(\log M)$ factor, where $M$ is an upper bound on the number of copies.


\paragraph{$k$-means clustering.}
Now we introduce the \emph{$k$-means clustering problem} and the notion of \emph{coreset}. 

\begin{definition}[$k$-means clustering]
Given a point set $Q\subseteq [\Delta]^d$ and a parameter $k\in\mathbb{N}_+$ for the target number of centers, the goal of $k$-means clustering  is to find a set of $k$ points $Z\subseteq \mathbb{R}^d$ such that the objective function, 
$
\cost(Q, Z) := \sum_{q\in Q} \dist^2(q, Z),
$
 is minimized.
Each point in $Z$ is called a \emph{center}.
 $\OPT$ is defined to be the optimal cost of the $k$-means clustering problem.
\end{definition}

However, solving the $k$-means problem exactly is NP-hard~\citep{adhp09}. Oftentimes, we only need a good approximation.
For the purpose of finding an approximate solution, an important concept is \emph{coreset}, which is a small subset of (weighted) points whose $k$-means solution is a good approximate solution for the entire dataset.
The formal definition is the following:


\begin{definition}[Coreset for $k$-means]\label{def:coreset_kmeans}
Given $Q\subseteq [\Delta]^d$, $k\in\mathbb{N}_{+}$ and $\epsilon>0$,
a set of point-weight pairs $S=\{(s_1, w_1), (s_2, w_2), \ldots, (s_m,w_m)\} \subset [\Delta]^d \times \R_{>0}$ is an \emph{$\epsilon$-coreset} for $Q$, where $w_i$ is the weight of $s_i$, if $S$ satisfies
$$
\forall Z\subset \mathbb{R}^d, |Z|=k: \big|\cost(S, Z) -\cost(Q, Z)\big|\le \epsilon \cdot \cost(Q, Z),
$$
where
$
\cost(S, Z) :=\sum_{i=1}^m w_i\dist^2(s_i, Z).
$
The size of the coreset is $|S|$.
\end{definition}

The main problem studied in this paper is how to construct a small coreset for $k$-means over a dynamic data stream. The formal description is the following.

\begin{definition}[Coreset for $k$-means over a dynamic stream]
Given a point set $Q\subseteq [\Delta]^d$ described by a dynamic stream of operations (Definition~\ref{def:dynamic_stream}), a parameter $k\in\mathbb{N}_+$ for the target number of centers, and an error parameter $\varepsilon\in(0,0.5)$. The goal is to design an algorithm in the dynamic streaming model  which can with probability at least $0.9$ output a small size $k$-means $\varepsilon$-coreset  (Definition~\ref{def:coreset_kmeans}) for $Q$ using as small space as possible.
\end{definition}

\paragraph{Sensitivity sampling based coreset construction.}
Let us briefly review the coreset construction framework proposed by~\cite{fl11,bfl16}. Given $Q\subseteq [\Delta]^d$ and $k \in \mathbb N_+$, the \emph{sensitivity} of a point $p\in Q$ is defined as:
\begin{align*}
s(p) = \max_{Z \in \mathbb{R}^d, |Z| = k} \frac{\dist^2(p,Z)}{ \sum_{q\in Q} \dist^2(q,Z) }.
\end{align*}
The following theorem gives guarantee of a sensitivity sampling based coreset construction.
\begin{theorem}[\cite{fl11,bfl16}]\label{thm:fl11}
	Given a set of points $Q \subseteq [\Delta]^d$ and a parameter $k$, let $s(p)$ denote the sensitivity of each point $p\in Q$.
	For each $p\in Q$, let $s'(p)$ be an upper bound on the sensitivity of $p$, i.e., $s'(p)\ge s(p)$, and let $t'=\sum_{p\in Q} s'(p).$
	Consider a multiset $S$ of $m$ i.i.d. samples from $Q$, where each sample chooses $p\in Q$ with probability $s'(p)/t'$. For each sampled point $p$, a weight $w(p)\in (1\pm \varepsilon/2)\cdot t'/(ms'(p))$ is associated with $p$. If $m \ge \Omega(t'\varepsilon^{-2}(\log|Q|\log t'+ \log (1/\delta)) )$,
	then with probability at least $1-\delta$, $\{(p,w(p))\mid p\in S\}$ is an $\epsilon$-coreset (Definition~\ref{def:coreset_kmeans}) for $Q$.  
\end{theorem} 
According to the above theorem, if we can find a good sensitivity upper bound $s'(p)$ for each point $p$, then we are able to construct a coreset with size nearly linear in $t'=\sum_{p} s'(p)$. In section~\ref{sec:off_sens}, we give an offline algorithm which can estimate a good sensitivity upper bound for each point, which readily implies an efficient offline coreset construction algorithm. In section~\ref{sec:streaming}, we show how to implement this sensitivity sampling procedure over a dynamic stream.
Notice that \cite{bfl16} gave a sensitivity sampling framework that works for clustering with general loss functions, and our method can be extended to those problems as well.

\vspace{-0.15in}
\section{An Offline Sensitivity Sampling Procedure}\label{sec:off_sens}
\vspace{-0.05in}
In this section, we consider the \emph{offline} setting in which all the data points are given. In this setting, we design a coreset construction algorithm based on sensitivity sampling.
In Section~\ref{sec:streaming}, we will show how to implement this algorithm in the dynamic streaming setting. 
\vspace{-0.15in}
\subsection{Randomly Shifted Grids}
\vspace{-0.05in}
We consider data points from $[\Delta]^d$ and assume without loss of generality that $\Delta=2^{L}$ for some positive integer $L$.
The space $[\Delta]^d$ is partitioned by a hierarchical grid structure as follows.
The first level (level $0$) of the grid contains  \emph{cells} with side-length $\Delta$ such that all the data points are contained in a single cell.
For each higher level, we refine the grid by splitting each cell into $2^d$ equal sized sub-cells.
In the finest level, i.e., the $L$-th level, each cell contains a single point.
We further randomly shift the boundary of the grids to achieve certain properties, which we will show later.
Formally, our grid structure is defined as the following.

\begin{definition}[Grids and cells] \label{def:grid_structure}
Let $g_0 = \Delta$. Choose a vector $v$ uniformly at random from $[0,\Delta]^d$. Partition the space $\mathbb{R}^d$ into a regular Cartesian grid $G_0$ with side-length $g_0$ and translate $G_0$ such that a vertex of this grid falls on $v$. The grid $G_0$ can be regarded as an infinite set of disjoint \emph{cells}, where each cell $C \in G_0$ can be expressed as
\begin{align*}
[v_1 + n_1 g_0,\ v_1 + (n_1+1) g_0) \times \cdots \times [ v_d + n_d g_0,\ v_d + (n_d+1) g_0)\subset \mathbb{R}^d
\end{align*}
for some $(n_1,n_2,\ldots,n_d)\in \mathbb{Z}^d$. (Note that each cell is a Cartesian product of intervals.)
	
For $i\geq 1$, we define the regular grid $G_i$ as the grid with side-length $g_i = g_0 /2^i$ aligned such that each cell in $G_{i-1}$ contains $2^d$ cells in $G_i$. The finest grid is $G_L$ where $L =  \log_2 \Delta $. A cell of $G_L$ has side-length $1$ and thus contains at most one data point.
\end{definition}


For convenience, we also define $G_{-1}$ to be the regular grid with side-length $g_{-1}=2 \Delta$, and each cell in $G_{-1}$ is a union of $2^d$ cells in $G_0$. Since the data points are in $[\Delta]^d$, there must be a single cell in $G_{-1}$ which contains all the data points. 
Consider two cells $C\in G_{i}$ and $C'\in G_j$ for some $i,j\in\{-1,0,1,\ldots,L\}$. If $C'\subset C$, then $C$ is an \emph{ancestor} of $C'$. Furthermore, if $i=j-1,$ then $C$ is the \emph{parent} of $C'$ and $C'$ is a \emph{child} of $C$. Thus every cell which is not from $G_L$ has exactly $2^d$ children cells. For a point $p$ (or a set $P$ of points), $c_i(p)$ (or $c_i(P)$) denotes the cell $C$ in grid $G_i$ which contains $p$ (or $P$). If $i$ is clear from the context, we will just use $c(p)$ (or $c(P)$) for short.

\subsection{Sensitivity Estimation and Coreset Construction}\label{sec:offline_sensitivity_estimation_and_coreset_construction}
In Algorithm~\ref{alg:offline_sens} we describe how to assign a sensitivity upper bound for every point. It needs an estimate $o$ of the optimal $k$-means cost $\OPT$. We will show how to enumerate the guesses $o$ later. According to Theorem~\ref{thm:fl11}, it directly gives an offline coreset construction algorithm.

\begin{algorithm}[t!]
	\small
	\begin{algorithmic}[1]\caption{Sensitivity Estimation}\label{alg:offline_sens}
		\STATE {\bfseries predetermined:} a guess $o\in [1,\Delta^d\cdot d\Delta]$ of the optimal $k$-means cost $\OPT$
		\STATE {\bfseries input:} a point set $Q\subseteq [\Delta]^d$, a parameter $k\in\mathbb{N}_{+}$
		\STATE Impose randomly shifted grids $G_{-1},G_0,G_1,\ldots,G_L$ (Definition~\ref{def:grid_structure}).
		\STATE Let $C\in G_{-1}$ be the cell which contains $Q$, i.e., $C= c(Q)$. Mark $C$ as \emph{heavy}.
		\FOR{$i:=0\rightarrow L-1$}
			\STATE Set the threshold value $T_i(o) = (d/g_i)^2\cdot o/k \cdot 1/100$.
			\FOR{$C\in G_i$ with $C\cap Q\not = \emptyset$}
				\STATE Let $z$ be an estimated value of $|C\cap Q|$ up to some precision. \label{sta:estimation_cell_size}
				\STATE If $z\geq T_i(o)$, mark $C$ as heavy. 
				\STATE Otherwise, if all the ancestors of $C$ are marked as heavy, mark $C$ as \emph{crucial}.
			\ENDFOR
		\ENDFOR
		\STATE For $C\in G_L$, if all the ancestors of $C$ are marked as heavy, mark $C$ as crucial.
		\STATE Initialize $Q_0=Q_1=\cdots=Q_L=\emptyset$.
		\STATE For $p\in Q,$ if $c_i(p)$ is marked as crucial, add $p$ into set $Q_i$ and set $s'(p)=10d^3/T_i(o)$. 
		\STATE {\bfseries output:} $Q_0,Q_1,\ldots,Q_L$ and $s'(\cdot)$
	\end{algorithmic}
\end{algorithm}

We also give 
 an alternative sampling procedure in Algorithm~\ref{alg:offline_coreset} which is useful for the dynamic streaming model. 

\begin{algorithm}[t!]
	\small
	\begin{algorithmic}[1]\caption{Sensitivity Sampling Based Coreset Construction}\label{alg:offline_coreset}
		\STATE {\bfseries predetermined:} a guess $o$ of the optimal $k$-means cost $\OPT$, an error parameter $\varepsilon\in(0,0.5)$
		\STATE {\bfseries input:} a point set $Q\subset [\Delta]^d$, a parameter $k\in \mathbb{N}_{+}$
		\STATE Let $Q_0,Q_1,\ldots,Q_L$ and $s'(\cdot)$ be the output of Algorithm~\ref{alg:offline_sens}.
		\STATE Let $\hat{q}_0,\hat{q}_1,\ldots,\hat{q}_L$ be the estimated values of $|Q_0|,|Q_1|,\ldots,|Q_L|$ respectively. \label{sta:estimated_level_size}
		\STATE For $i\in \{0,1,\ldots,L\},$ set $T_i(o)=(d/g_i)^2 \cdot o/k\cdot 1/100$ (same as in Algorithm~\ref{alg:offline_sens}).
		\STATE Set $\gamma = \varepsilon/(40^2 L d^3)$.
		\STATE Let $I=\{i\mid 0\le i\le L, \hat{q}_i\ge \gamma T_i(o) \}$.\label{sta:set_important_level_I}
		\STATE Let $Q^I=\bigcup_{i\in I} Q_i$. \label{sta:definition_of_QI}
		\qquad\qquad\qquad\qquad\codecmt{Only consider the levels with sufficient number of points.}
		\STATE Set $t'=\sum_{i\in I} \hat{q}_i \cdot 10d^3/T_i(o)$.  \qquad\qquad\qquad\qquad\qquad\quad \codecmt{Total estimated sensitivities.}
		\STATE Set $m=\Theta(t'\varepsilon^{-2}Ld\log t')$ and initialize $S=\emptyset$. \qquad\qquad \codecmt{$m$ is the total number of samples.}
		\FOR{$j=1\rightarrow m$}
			\STATE Choose a random level $i\in I$ with probability $(\hat{q}_i\cdot  10d^3/T_i(o)) / t'$.
			\STATE Uniformly sample a point $p$ from $Q_i$.
			\STATE Add $(p,t'/(ms'(p)))$ to set $S$.
		\ENDFOR
		\STATE {\bfseries output:} the set $S$
	\end{algorithmic}
\end{algorithm}

\begin{theorem}\label{thm:offline_sampling_coreset}
	Suppose that for any $ i\in\{0,1,\ldots,L\}$ and for any cell $C\in G_i$ with $C\cap Q\not=\emptyset$, the estimated value $z$ in line~\ref{sta:estimation_cell_size} of Algorithm~\ref{alg:offline_sens} satisfies either $z\in |C\cap Q|\pm 0.1 T_i(o)$ or $z\in(1\pm 0.01) \cdot |C\cap Q|$, and for any $Q_i$, the estimated value $\hat{q}_i$ in line~\ref{sta:estimated_level_size} of Algorithm~\ref{alg:offline_coreset} satisfies either $\hat{q}_i \in |Q_i|\pm 0.1\varepsilon\gamma T_i(o)$ or $\hat{q}_i \in (1\pm 0.01\epsilon)\cdot|Q_i| $.
Suppose $o\in (0,\OPT]$.
 Then the set $S$ output by Algorithm~\ref{alg:offline_coreset} is an $\varepsilon$-coreset (Definition~\ref{def:coreset_kmeans}) for $Q$. Furthermore, with probability at least $0.93$, $|S|$ is at most $O(k\varepsilon^{-2}d^4L^2\log(kdL)\cdot (\OPT/o + 1))$.
\end{theorem}

\subsection{Analysis}\label{sec:offline_analysis}
 Now we give the proof of Theorem~\ref{thm:offline_sampling_coreset}.
 All the missing proofs in this section are given in Appendix~\ref{app:missing_proofs_offline}.
Let us first state some simple facts.
\begin{fact}\label{fac:Q_is_partitioned}
	The point sets $Q_0,Q_1,\ldots,Q_L$ obtained by Algorithm~\ref{alg:offline_sens} form a partition of $Q$, i.e., for all $ p\in Q,$ there is exactly one $i\in\{0,1,\ldots, L\}$ such that $p\in Q_i$.
\end{fact}

\begin{fact} \label{fact:trivial-1}
	For $ C\in G_i$, if $C$ is marked as heavy, then $|C\cap Q|\geq 0.9 T_i(o)$; otherwise $|C\cap Q|\leq 1.1 T_i(o)$. Similarly, if $i\in I$, then $|Q_i|\geq 0.9\gamma T_i(o)$; otherwise $|Q_i|\leq 1.1\gamma T_i(o)$.
\end{fact}

\begin{fact} \label{fact:trivial-2}
For $Q_i$ output by Algorithm~\ref{alg:offline_sens}, every point $p\in Q_i$ is assigned the same sensitivity upper bound $s'(p)=10d^3/T_i(o)$.
\end{fact}

In line~\ref{sta:set_important_level_I} of Algorithm~\ref{alg:offline_coreset}, we set $I$ to be the set of levels such that there are sufficient number of points in the crucial cells in those levels. 
The following lemma shows that the point set
$Q^I$ (line~\ref{sta:definition_of_QI} of Algorithm~\ref{alg:offline_coreset}) is 
a good representative of the point set $Q$, i.e., for any set of $k$ centers $Z$, the $k$-means cost $\cost(Q,Z)$ is close to the $\cost(Q^I,Z)$.

\begin{lemma}\label{lem:reducing_to_important_points}
Let $Q^I$ and $\varepsilon$ be the same as in Algorithm~\ref{alg:offline_coreset}. If $o\in(0,\OPT]$, then for any $Z\subseteq \mathbb{R}^d$ with $|Z|=k$, we have:
$
\cost(Q^I,Z)\leq \cost(Q,Z)\leq (1+\varepsilon/10) \cost(Q^I,Z).
$
\end{lemma}
Next, instead of showing $s'(p)$ (output by Algorithm~\ref{alg:offline_sens}) is a sensitivity upper bound with respect ot $Q$, we show that $s'(p)$ is also an sensitivity upper bound with respect to $Q^I$. This is even stronger since $Q^I$ is a subset of $Q$ and we have:
\vspace{-0.1in}
\begin{align*}
\forall Z\subseteq \mathbb{R}^d :|Z|=k,\ \frac{\dist^2(p,Z)}{\sum_{q\in Q} \dist^2(q,Z)}\leq \frac{\dist^2(p,Z)}{\sum_{q\in Q^I}\dist^2(q,Z)}.
\end{align*}
\vspace{-0.2in}

\begin{lemma}\label{lem:sensitivity_ub}
Let $Q^I$ be the same as in Algorithm~\ref{alg:offline_coreset}. If $o\in (0,\OPT]$, then for all $ i\in\{0,1,2,\ldots,L\}$ and $ p\in Q_i $, we have:
\vspace{-0.1in}
\begin{align*}
\max_{Z\subseteq \mathbb{R}^d:|Z|=k}\frac{\dist^2(p,Z)}{\sum_{q\in Q^I}\dist^2(q,Z)} \leq 10\frac{d^3}{T_i(o)} = s'(p).
\end{align*}
\end{lemma}

Now, we explain the reason of imposing randomly shifted grids. We fix an optimal set $Z^*=\{z_1^*,z_2^*,\ldots,z_k^*\}\subset \mathbb{R}^d$ of $k$ centers for the point set $Q$, i.e., $\cost(Q,Z^*)=\OPT$. We call a cell $C\in G_i$ a \emph{center cell} if it is close to a center in $Z^*$, namely $\dist(C,Z^*)\leq g_i/(2d)$. We claim that there will not be too many center cells since we randomly shift the grids. In other words, each center of $Z^*$ is far from the boundary of every gird.

\begin{lemma}\label{lem:num_center_cell}
	With probability at least $0.94$, the total number of center cells is at most $100 kL$. 
\end{lemma}
This lemma is similar to Lemma 2.2 in \citep{bflsy17}.
For completeness, we also provide a proof in Appendix~\ref{app:missing_proofs_offline}.


 Consider the total estimated sensitivities, i.e., the sum of the sensitivity upper bounds over all the points. Due to Theorem~\ref{thm:fl11}, this sum determines the size of the coreset. We show that if the estimate $o$ of the optimal $k$-means cost is close to $\OPT$, then the total estimated sensitivities can not be too large.

\begin{lemma}\label{lem:total_sensitivity_ub}
Suppose the number of center cells is at most $100kL$.
Let $Q_0,Q_1,\ldots,Q_L,s'(\cdot)$ be the output of Algorithm~\ref{alg:offline_sens}. Then the total estimated sensitivities satisfies
$
\sum_{p\in Q} s'(p)\le  4000kLd^3   \cdot \left(\OPT/o + 1\right).
$
\end{lemma}

Since $Q^I$ is a subset of $Q$, according to the above lemma, we also have $\sum_{p\in Q^I} s'(p)\le 4000d^3 L k  \cdot \left(\OPT/o + 1\right) $.
Now, we are ready to prove Theorem~\ref{thm:offline_sampling_coreset}.

\begin{proofof}{Theorem~\ref{thm:offline_sampling_coreset}}
By Lemma~\ref{lem:num_center_cell}, with probability at least $0.94$, the number of center cells is at most $100kL$. 
In the following, we condition on this event. 

Algorithm~\ref{alg:offline_coreset} draws $m$ i.i.d. samples. For each sample, a point $p\in Q_i\subseteq Q^I$ is chosen with probability 
\vspace{-0.2in}
\begin{align*}
\frac{\hat{q}_i/T_i(o)}{\sum_{j\in I} \hat{q}_j/T_j(o)} \cdot \frac{1}{|Q_i|} = \frac{\frac{\hat{q}_i}{|Q_i|}\cdot 20\frac{d^3}{T_i(o)}}{\sum_{j\in I}\sum_{p'\in Q_j} \frac{\hat{q}_j}{|Q_j|}\cdot 20\frac{d^3}{T_j(o)}}.
\end{align*}
\vspace{-0.1in}
Each sample $p$ is given a weight 
\begin{align*}
\frac{t'}{ms'(p)} = \frac{1}{m} \cdot \frac{\sum_{j\in I}\sum_{p'\in Q_j} \frac{\hat{q}_j}{|Q_j|}\cdot 20\frac{d^3}{T_j(o)}}{\frac{\hat{q}_i}{|Q_i|}\cdot 20\frac{d^3}{T_i(o)}}\cdot \frac{\hat{q}_i}{|Q_i|}\in (1\pm \varepsilon/4) \cdot \frac{1}{m}\cdot \frac{\sum_{j\in I}\sum_{p'\in Q_j} \frac{\hat{q}_j}{|Q_j|}\cdot 20\frac{d^3}{T_j(o)}}{\frac{\hat{q}_i}{|Q_i|}\cdot 20\frac{d^3}{T_i(o)}}.
\end{align*}
Let $s''(p)=\frac{\hat{q}_i}{|Q_i|}\cdot 20\frac{d^3}{T_i(o)}$. Since $\hat{q_i}/|Q_i|\geq 1/2,$ we know that $s''(p)$ is still a sensitivity upper bound of $p$ with respect to $Q^I$ by Lemma~\ref{lem:sensitivity_ub}.
According to Algorithm~\ref{alg:offline_coreset}, we have $t'=\frac{1}{2}\sum_{p\in Q^I} s''(p)$. Therefore, if we set $m$ to be a sufficiently large $\Omega(t'\varepsilon^{-2}Ld\log t')=\Omega(t'\varepsilon^{-2}(\log(\Delta^d)\log t'+\log (1/0.01)))$, then according to Theorem~\ref{thm:fl11}, $S$ output by Algorithm~\ref{alg:offline_coreset} is an $\varepsilon/2$-coreset for $Q^I$ with probability at least $0.99$. By Lemma~\ref{lem:reducing_to_important_points}, if $S$ is an $\varepsilon/2$-coreset for $Q^I$, then $S$ is also an $\varepsilon$-coreset for $Q$. Thus, the correctness is proved, and the overall success probability is at least $0.93$ obtained by a simple union bound.
Now let us analyze the size of the coreset $S$.
  Since $\forall j\in I, \hat{q}_j/|Q_j|\leq 2,$ we have $\sum_{j\in I}\sum_{p'\in Q_j} s''(p')\leq 2t'\leq 4\cdot 4000 d^3Lk\cdot(\OPT/o + 1)$ by Lemma~\ref{lem:total_sensitivity_ub}. Thus, the size of $S$ is $m=O(k\varepsilon^{-2}d^4L^2\log(kdL)\cdot (\OPT/o + 1))$.
\end{proofof}

\vspace{-0.15in}
\section{Coreset Construction over a Dynamic Stream}\label{sec:streaming}
In this section, we show how to implement Algorithms~\ref{alg:offline_sens}~and~\ref{alg:offline_coreset} in the dynamic streaming setting.
We defer all the missing details in this section to Appendix~\ref{sec:missing_details_of_dynamic_stream}.

First, we introduce a dynamic storage structure that allows us to insert and delete points or cells.
We then use this data structure combined with hash functions to estimate the number of points falling into each cell.
Lastly, we combine them with the sensitivity sampling procedure to obtain our final algorithm.

\begin{algorithm}[h!]
	\small
	\begin{algorithmic}[1]\caption{Point-cell storing procedure}\label{alg:storing}
		\STATE \textsc{Storing}$(G_i,\alpha,\beta,\delta)$: 
		\STATE {\bfseries Input:} $\{((p_1,l_1),\pm),((p_2,l_2),\pm),\ldots\}$ \qquad\codecmt{
			$l_t\in[\hat{m}]$. Only Algorithm~\ref{alg:sensitivity_sampler} uses the case for $\hat{m}>1$.}
		\STATE Run \textsc{Distinct}$(\alpha,\delta/4)$  on $\{(c_i(p_1),\pm),(c_i(p_2),\pm),\ldots\}$ in parallel. \label{sta:global_estimation}
		\STATE Set $r=\lceil\log(4\alpha/\delta)\rceil$ and $h_1,h_2,\ldots,h_r$, $\forall j\in[r],$ $h_j:G_i\rightarrow [2 \alpha]$. ~~~ \codecmt{$h_j$ is pairwise independent.}
		\STATE Run $r\cdot 2 \alpha$ copies of \textsc{Distinct}$(\beta,\delta/(2\alpha ))$ in parallel.
		\label{sta:multi_copy_distinct}\\
		\qquad\codecmt{	
			Each copy is indexed by a pair $(j,b)\in [r]\times [2 \alpha]$. The $(j,b)$-th copy 
			is run on the sub-stream $\{((p'_1,l'_1),\pm),((p'_2,l'_2),\pm),\ldots\}$, where each $((p'_t,l'_t),\pm)$ satisfies $h_j(c_i(p'_t)) = b$.}
		\STATE If line~\ref{sta:global_estimation} returns \textbf{FAIL}, output \textbf{FAIL}; otherwise, let $\mathcal{C},f:\mathcal{C}\rightarrow \mathbb{N}_{+}$ be the output of line~\ref{sta:global_estimation}.\\
		\qquad\qquad\qquad\qquad
		\codecmt{$\mathcal{C}\subset G_i$ contains all the cells found, and 
			$f(C)$ denotes the number of points in $C$.}
		\STATE Initialize $S\gets\emptyset$.
		\FOR{$C\in\mathcal{C}$ with $f(C)\leq \beta$}
		\STATE Find $j\in[r]$ s.t. 
		$\forall C'\in\mathcal{C},h_j(C)\not = h_j(C')$
		and
		the $(j,h_j(C))$-th copy 
		in line~\ref{sta:multi_copy_distinct} does not \textbf{FAIL}.
		\STATE If such $j$ does not exist, output \textbf{FAIL};
		otherwise $S\leftarrow S\cup S_{j,h_j(C)}$.\\
		\qquad\qquad\codecmt{Here $S_{j,h_j(C)}$ is the set of distinct points found by the $(j,h_j(C))$-th copy 
			in line~\ref{sta:multi_copy_distinct}.}
		\ENDFOR
		\STATE {\bfseries Output:} $\mathcal{C}\subset G_i,f:\mathcal{C}\rightarrow \mathbb{N}_+,S=\{(\tilde{p}_1,\tilde{l}_1),(\tilde{p}_2,\tilde{l}_2),\ldots\}$
	\end{algorithmic}
\end{algorithm}

\paragraph{The dynamic point-cell storing data structure.}
We introduce an algorithm that maintains a set of points and cells in a dynamic data stream. Before that, let us recall \cite{ganguly2005counting}'s result for finding distinct elements, which we use as a subroutine in our algorithm.

\begin{lemma}[Distinct elements~\citep{ganguly2005counting}]\label{lem:k_set}
Given parameters $M \geq 1,N \geq 1,s \geq 1,\delta \in (0,1/2)$, there is an algorithm \textsc{Distinct}$(s,\delta)$ that requires $O(s ( \log M + \log N) \log (s/\delta))$ bits to process a stream of insertion/deletion of data items. For each operation $(i,\pm)$ ($i\in[N]$), the algorithm takes $O(\log(s/\delta))$ time. $M$ is an upper bound of the total frequency of all items during the stream. At the end of the stream, if the number of distinct elements is at most $s$, with probability at least $1-\delta$ it returns all the distinct elements and their  frequencies. It returns \textbf{FAIL} otherwise.
\end{lemma}


We use \textsc{Distinct} as our sub-routine. We  set the parameter $M$ and $N$ to be sufficiently large in our case, i.e., $M=N=\Delta^{2d}$.
In Algorithm~\ref{alg:storing}, we describe a method which can with probability at least $1-\delta$ output all the non-empty cells in grid $G_i$ when the total number of non-empty cells is not too large (at most $\alpha$). Furthermore, if the number of points in a particular cell is not too large (at most $\beta$), the algorithm can output all the points in that cell. Notice that Algorithm~\ref{alg:storing} is only a subroutine of our final algorithm and will only work on some sub-stream of the entire data stream.

\paragraph{Estimating the number of points in each cell.}
We use Algorithm~\ref{alg:storing} as a subroutine and design a dynamic streaming algorithm (Algorithm~\ref{alg:streaming_est_cell}) that can estimate the number of points in each cell up to some precision. Furthermore, it also estimates the number of points $|Q_i|$ in crucial cells of each level $i$.

\begin{algorithm}[t!]
	\small
	\begin{algorithmic}[1]\caption{Estimating the number of points in each cell and in each level}\label{alg:streaming_est_cell}
		\STATE \textsc{PointsEstimation}$(o,\varepsilon,\delta)$:
		\STATE \textbf{Input:} a point set $Q\subseteq [\Delta]^d$ described by a stream $\{(p_1,\pm),(p_2,\pm),\ldots\}$ 
		\FOR{$i\in\{0,1,\ldots,L\}$,  in parallel, } 
		\STATE $T_i(o)\gets (d/g_i)^2 \cdot o/(100k)$, 
		$\alpha\gets10^{11} kLd\log(1/\delta)$,
		$\alpha'\gets10^{16}\varepsilon^{-3}kL^2d^4$.\\
		\qquad\qquad\codecmt{$T_i(o)$: threshold for heaviness (Algortihm~\ref{alg:offline_sens});\qquad $\alpha,\alpha'$: parameters for \textsc{Storing}.}
		\STATE 
		Let $h_i:[\Delta]^d\rightarrow \{0,1\}$ 
		be a $\lambda$-wise independent hash function.
		\\
		\qquad\qquad
		\codecmt{$\lambda = 10\lceil(dL+\log(1/\delta)+1)\rceil$;\qquad $\forall p\in[\Delta]^d, h_i(p)=1$ w.p. $\min(4\cdot10^4\lambda/T_i(o),1)$.
		}
		\STATE 
		Run \textsc{Storing}$(G_i,\alpha,1,0.1\delta/L)$  on a sub-stream $\{((p'_1,1),\pm),((p'_2,1),\pm),\ldots\}$ in parallel.
		\\
		\qquad\qquad
		\codecmt{\textsc{Storing} is defined in
			Algorithm~\ref{alg:storing}. Here $p'_j$ satisfies $h_i(p'_j) = 1$.} \\
		If \textsc{Storing} returns \textbf{FAIL}, output \textbf{FAIL}. \label{sta:run_storing}
		\STATE Let $\gamma\gets \varepsilon/(40^2Ld^3)$.
		\qquad\codecmt{Threshold for discarding levels}
		\STATE 
		Let  $h'_i:[\Delta]^d\rightarrow \{0,1\}$ be a $\lambda$-wise independent hash function.  \\
		\qquad\qquad\codecmt{
			$\forall p\in [\Delta]^d,h'_i(p)=1$ w.p. $\min(4\cdot 10^4\varepsilon^{-2}\gamma^{-1}\lambda/T_i(o),1)$.
		}
		\STATE 
		Run \textsc{Storing}$(G_i,\alpha',1,0.1\delta/L)$
		on sub-stream $\{((p''_1,1),\pm),((p''_2,1),\pm),\ldots\}$ in parallel. \\
		\qquad\qquad
		\codecmt{\textsc{Storing} is defined in Algorithm~\ref{alg:storing}.  Here $p''_j$ satisfies $h_i'(p''_j)=1$.} \\
		If \textsc{Storing} returns \textbf{FAIL}, output \textbf{FAIL}. \label{sta:call_storing_dominating}
		\STATE 
		Let $\mathcal{C}_i,f_i,S_i$ $\gets$ \textsc{Storing}$(G_i,\alpha,1,0.1\delta/L)$ in line~\ref{sta:run_storing}. Let $\hat{f}(C)=f_i(C)\cdot\min\left\{T_i(o)/(4\cdot 10^4 \lambda),1\right\}$.\\
		 \qquad\qquad\codecmt{
			Use $\hat{f}(C)$ as an estimator for $|C\cap Q|$ and follow Algorithm~\ref{alg:offline_sens} to determine whether $C$ is marked as heavy, crucial or nothing. (And conceptually compute $Q_0,\ldots,Q_L$ for analysis.)
		}
		\STATE 
		Let $\mathcal{C}_i',f_i',S_i'$ $\gets$ \textsc{Storing}$(G_i,\alpha',1,0.1\delta/L)$ in line~\ref{sta:call_storing_dominating}.
		\STATE Let 
		$\hat{q}_i = \min\left\{\varepsilon^2\gamma T_i(o)/(4\cdot 10^4\lambda) \cdot \sum_{C\in\mathcal{C}_i':C\text{ is crucial}} f_i'(C),1\right\} $.
		\ENDFOR
		\STATE {\bfseries Output:} $\hat{q}_0,\hat{q}_1,\ldots,\hat{q}_L$ and $\hat{f}:\bigcup_{i=0}^L G_i \rightarrow \mathbb{R}_+$
	\end{algorithmic}
\end{algorithm}

\paragraph{Sensitivity sampling over a dynamic stream.}
Since using Algorithm~\ref{alg:streaming_est_cell} we can estimate the number of points in each cell and the size of each $Q_i$, the only remaining thing for simulating Algorithm~\ref{alg:offline_coreset} is to draw samples based on their sensitivity upper bounds. 
In Algorithm~\ref{alg:sensitivity_sampler}, we show how to achieve this in a dynamic stream.

\begin{algorithm}[h!]
	\small
	\begin{algorithmic}[1]\caption{Sensitivity sampling over a dynamic stream}\label{alg:sensitivity_sampler}
		\STATE \textsc{Sampling}$(o,\varepsilon,\delta)$:
		\STATE \textbf{Input:} a point set $Q\subseteq [\Delta]^d$ described by a stream
		$\{(p_1,\pm),(p_2,\pm),\ldots\}$
		\STATE Run \textsc{PointsEstimation}$(o,\varepsilon,\delta/2)$ in parallel on the input stream. If it returns \textbf{FAIL}, output \textbf{FAIL}.
		 \label{sta:run_dynamic_partition}
		\STATE $\forall i\in\{0,1,\ldots,L\}$, $T_i(o) \gets (d/g_i)^2 \cdot o/k\cdot 1/100$.
		\qquad\qquad\codecmt{Threshold for heaviness (Algorithm~\ref{alg:offline_sens}).}
		\STATE $\hat{m}\gets \Theta\left(k\varepsilon^{-3} L^4d^7\log\left(\frac{dLk}{\delta}\right)\cdot\frac{1}{\delta}\right)$. 
		 \qquad\qquad \codecmt{$\wh{m}$ hash functions needed for $m$ independent samples. 
		}
	\STATE $\lambda \gets 10\lceil(dL+\log(1/\delta)+1)\rceil$. \qquad\qquad~~~~~~\codecmt{$\lambda$ is the independence parameter.}
		\STATE  $\forall i\in\{0,1,\ldots,L\},$ choose $h_{i,1},h_{i,2},\ldots,h_{i,\hat{m}}:[\Delta]^d\rightarrow\{0,1\}$.\\ 
		\qquad~~~~
		\codecmt{
			$\forall j\in[\hat{m}],$ $h_{i,j}$ is $\lambda$-wise independent, $\forall p\in [\Delta]^d$, $h_{i,j}(p)=1$ w.p. $\min\left\{1/(10^4 kL T_i(o)),1\right\}$. 
		}
		\STATE  $\alpha\gets\Theta( k\varepsilon^{-3}L^3d^7\log(dLk/\delta)\delta^{-1}),\beta\gets \Theta(\varepsilon^{-3}L^3d^7\log(dLk/\delta)\delta^{-1})$.
		\qquad~~~~~\codecmt{$\alpha, \beta$:  for \textsc{Storing}.}
		\STATE For $i\in\{0,1,\ldots,L\}$, 
	 run \textsc{Storing}$(G_i,\alpha,\beta,0.1\delta/L)$ (Algorithm~\ref{alg:storing}) in parallel. \\
		\qquad~~~
		\codecmt{Each instance is run on a new stream obtained by splitting each operation $(p_t,\pm)$ from the original input stream into a set of new operations $\{((p_t,j),\pm)\mid j\in [\hat{m}],h_{i,j}(p_t)=1\}$. 
		}\\
		 If any \textsc{Storing} returns \textbf{FAIL}, output \textbf{FAIL}. \label{sta:store_samples}
		\STATE  $\hat{q}_0,\hat{q}_1,\ldots,\hat{q}_L$, $\hat{f}:\bigcup_{i=0}^L G_i\rightarrow \mathbb{R}_+$ \qquad$\gets$\qquad \textsc{PointsEstimation}$(o,\varepsilon,\delta/2)$ in line~\ref{sta:run_dynamic_partition}.
		\STATE		 $\mathcal{C}_i,f_i,S_i$ $\gets$ \textsc{Storing}$(G_i,\alpha,\beta,0.1\delta/L)$ in line~\ref{sta:store_samples}.
		\STATE For $i\in\{0,1,\ldots,L\}$,
		 if $\exists C\in\mathcal{C}_i$ marked as crucial 
		 in line~\ref{sta:run_dynamic_partition}, and $f_i(C)>\beta$, output \textbf{FAIL}. \label{sta:guarantee_storing_correctly}
		\STATE  $\gamma\gets \varepsilon/(40^2Ld^3)$. \qquad\quad\qquad\qquad\qquad~~ \codecmt{Threshold for discarding levels.}
		\STATE  $I\gets \{i\mid 0\leq i\leq L,\hat{q}_i\geq \gamma T_i(o)\}$. 
		\STATE  $t'\gets \sum_{i\in I} \hat{q}_i\cdot 10d^3/T_i(o)$.\qquad\qquad\qquad \codecmt{Total estimated sensitivities.}
		\STATE  $m\gets\Theta(t'\varepsilon^{-2} Ld \log (t'/\delta))$,  $S\gets\emptyset$.\qquad~ \codecmt{Total number of samples.}
		\STATE For $i\in\{0,1,\ldots,L\},$ set $A_i=[\hat{m}]$.
		\FOR {$j=1\rightarrow m$}\label{sta:uniform_sample_stage}
		\STATE Choose a random level $i\in I$ with probability $(\hat{q}_i\cdot 10d^3/T_i(o))/t'$. \label{sta:choose_i}
		\STATE Choose the minimum $j\in A_i$ s.t. $\exists p\in Q_i$, $(p,j)\in S_i$. If no such $j$, output \textbf{FAIL}. \label{sta:find_a_sample}
		\STATE Uniformly choose a point $p$ from the set $\{q\in Q_i\mid (q,j)\in S_i\}$. \label{sta:uniform_sampling_doable}
		\STATE Update $A_i\leftarrow \{j+1,j+2,\ldots,\hat{m}\}$.
		\STATE Add $(p,t'/(ms'(p)))$ to set $S$. \qquad\qquad\codecmt{$s'(p)=10d^3/T_i(o)$.} 
		\ENDFOR
		\STATE \textbf{\bfseries Output:} the set $S$
	\end{algorithmic}
\end{algorithm}


\paragraph{The final algorithm.} Finally, we use exponential search to enumerate the guesses $o$. In Algorithm~\ref{alg:final}, we show the details of how to run Algorithm~\ref{alg:sensitivity_sampler} with different guesses in parallel. Our main theorem is the following.

\begin{algorithm}[t]
	\small
	\begin{algorithmic}[1]\caption{Coreset construction over a dynamic stream}\label{alg:final}
		\STATE \textsc{DynamicCoreset}$(\varepsilon)$:
		\STATE \textbf{input:} a point set $Q\subseteq [\Delta]^d$ described by a stream
		$\{(p_1,\pm),(p_2,\pm),\ldots\}$
		\STATE Impose randomly shifted grids $G_{-1},G_0,G_1,\ldots,G_L$ (Definition~\ref{def:grid_structure}).
		\STATE Run \textsc{Distinct}$(10000k,0.001)$ (Lemma~\ref{lem:k_set}) over the input stream in parallel. \label{sta:recover_all}
		\STATE If line~\ref{sta:recover_all} does not output \textbf{FAIL}, we output the entire point set $Q$.
		\STATE For each $u\in [2dL]$, let $o_u=2^u\cdot 50k$ and run \textsc{Sampling}$(o_u,\varepsilon,0.001/(dL))$ (Algorithm~\ref{alg:sensitivity_sampler}) over the input stream in parallel. \label{sta:parallel_sampling}
		\STATE Set a threshold $h=\Theta(k\varepsilon^{-2}L^2d^4\log(kLd))$.
		\STATE Find the smallest $u^*$ such that \textsc{Sampling}$(o_{u^*},\varepsilon,0.001/(dL))$ in line~\ref{sta:parallel_sampling} does not output \textbf{FAIL}, and the returned set $S^*$ has size at most $h$, i.e., $|S^*|\leq h$.
		\STATE If no such $u^*$, output \textbf{FAIL}. Otherwise, output $S^*$ returned by \textsc{Sampling}$(o_{u^*},\varepsilon,0.001/(dL))$.
	\end{algorithmic}
\end{algorithm}

\begin{theorem}\label{thm:main}
Suppose a point set $Q\subseteq [\Delta]^d$ is given by a stream of insertion/deletion operations in the dynamic streaming model (Definition~\ref{def:dynamic_stream}). Let $L=\log \Delta$. For given $\varepsilon\in(0,1/2)$, Algorithm~\ref{alg:final} uses a single pass over the stream and on termination outputs a $k$-means $\varepsilon$-coreset $S$ (Definition~\ref{def:coreset_kmeans}) for $Q$ with probability at least $0.9$. Furthermore, the size of the coreset is at most $O(k\varepsilon^{-2}d^4L^2\log(kdL))$. The total space used by the algorithm is $\widetilde{O}(k)\cdot \poly(dL/\varepsilon)$ bits.
\end{theorem}


\section{Conclusion}

This paper gives the first $k$-means coreset construction in the dynamic streaming model using space polynomial in the dimension $d$ and nearly optimal (linear) in $k$.
The algorithm is based on sensitivity sampling, which we believe is a powerful tool and can have broader applications.


\bibliography{ref}

\appendix

\section{Missing Details in Section~\ref{sec:off_sens}}\label{app:missing_proofs_offline}

\begin{proofof}{Fact~\ref{fac:Q_is_partitioned}}
	Consider an arbitrary point $p\in Q.$ Let $C_{-1},C_0,\ldots,C_L$ be the cells which contain $p$, where $\forall i\in\{-1,0,1,\ldots, L\}$, the cell $C_i$ is from the grid $G_i$. According to Algorithm~\ref{alg:offline_sens}, $C_{-1}$ is marked as heavy and $C_L$ cannot be marked as heavy. Let $l$ be the largest integer such that all the cells $C_{-1},C_0,\ldots,C_{l-1}$ are marked as heavy. Then the cell $C_l$ must be marked as crucial, and all the cells $C_{l+1},C_{l+2},\ldots,C_{L}$ can not be crucial. Thus, we have $p\in Q_l$ and $\forall i\in\{0,1,\ldots,L\}\setminus\{l\},$ $p\not\in Q_i$.
\end{proofof}

Facts~\ref{fact:trivial-1}~and~\ref{fact:trivial-2} are obvious from the algorithms, so we omit the proofs.

The following claim is useful in the proofs.

\begin{claim}\label{cla:heavy_cell_contains_important_points}
	Let $Q^I$ and $\gamma$ be the same as mentioned in Algorithm~\ref{alg:offline_coreset}.
	For $ i\in \{0,1,\ldots,L\}$ and any heavy cell $C\in G_{i-1}$, 
	if all the ancestors of $C$ are marked as heavy, then we have $|C\cap Q^I|\geq (1-5(L-i)\gamma)\cdot |C\cap Q|$.
\end{claim}
\begin{proof}
	The proof is by induction. When $i=L$, consider a heavy cell $C\in G_{L-1}$ whose ancestors are also heavy. If there is no such cell $C$, then the claim holds directly for $i=L$. Otherwise, according to the construction of $Q_0,Q_1,\ldots,Q_L$, $\forall j\in\{0,1,\ldots,L-1\},$ we have $Q_j\cap C=\emptyset$ and $(Q\cap C) \subseteq Q_L$. Since $C$ is marked as heavy, we know that $|Q_L\cap C|=|Q\cap C|\geq 0.9T_{L-1}(o)\geq 0.9T_{L}(o)/4$. It implies that $\hat{q}_L\geq \min(|Q_L|- 0.1\varepsilon\gamma T_L(o),(1-0.01\varepsilon) |Q_L|)\geq \min(|Q\cap C|- 0.1\varepsilon\gamma T_L(o),(1-0.01\varepsilon) |Q\cap C|)\geq \gamma T_L(o)$. Thus, we have $Q_L\subseteq Q^I$ which implies that $|C\cap Q^I|=|C\cap Q_L|=|C\cap Q|$.
	
	
	
	Now assume the claim is true for $i+1,i+2,\ldots,L.$ Consider a heavy cell $C\in G_{i-1}$ whose ancestors are also marked as heavy. If there is no such cell $C$, the claim holds directly for $i$. Now consider the case when $C$ exists. If level $i\in I$, i.e., $Q_i\subseteq Q^I$, we have
	\begin{align*}
	|C\cap Q^I| &= \sum_{C'\in G_i:\text{$C'$ is a heavy child of $C$}} |C'\cap Q^I| + \sum_{C'\in G_i:\text{$C'$ is a crucial child of $C$}}|C'\cap Q^I|\\
	&= \sum_{C'\in G_i:\text{$C'$ is a heavy child of $C$}} |C'\cap Q^I| + |C\cap Q_i|\\
	&\geq (1-5(L-i-1)\gamma) \sum_{C'\in G_i:\text{$C'$ is a heavy child of $C$}} |C'\cap Q| + |C\cap Q_i|\\
	&\geq (1-5(L-i-1)\gamma)|C\cap Q|\\
	&\geq (1-5(L-i)\gamma)|C\cap Q|.
	\end{align*}
	If level $i\not \in I$, i.e., $Q_i\not \subseteq Q^I$, we have
	\begin{align*}
	\sum_{C'\in G_i:\text{$C'$ is a heavy child of $C$}} |C'\cap Q| \geq |C\cap Q| - |Q_i|\geq |C\cap Q| - 1.1\gamma T_i(o)\geq (1-5\gamma)|C\cap Q|.
	\end{align*}
	Thus,
	\begin{align*}
	|C\cap Q^I|&\geq \sum_{C'\in G_i:\text{$C'$ is a heavy child of $C$}} |C'\cap Q^I|\\
	&\geq (1-5(L-i-1)\gamma) \cdot \sum_{C'\in G_i:\text{$C'$ is a heavy child of $C$}} |C'\cap Q|\\
	&\geq (1-5(L-i-1)\gamma)(1-5\gamma)|C\cap Q|\\
	&\geq (1-5(L-i)\gamma) |C\cap Q|. 
	\end{align*}
\end{proof}

\begin{proofof}{Lemma~\ref{lem:reducing_to_important_points}}
	Since $Q^I$ is a subset of $Q$, $\cost(Q^I,Z)\leq \cost(Q,Z)$ is trivial. In the following, we are trying to prove $\cost(Q,Z)\leq (1+\varepsilon/10) \cost(Q^I,Z)$.
	
	We consider an level $i\not\in I$, i.e., $Q_i\not\subseteq Q^I$. For a point $p\in Q_i$, all the ancestors of $c_{i}(p)$ must be heavy. By averaging argument, there must exist a point $q\in c_{i-1}(p)\cap Q^I$ such that
	\begin{align}\label{eq:averaging_argument}
	\dist^2(p,Z) \leq & ~ 2 \dist^2(p,q) + 2 \dist^2(q,Z) \notag \\
	\leq & ~ 2 d g_{i-1}^2 + 2 \dist^2(q,Z) \notag \\
	\leq & ~ 2 d g_{i-1}^2 + 2 \frac{ 1 }{ | c_{i-1}(p)\cap Q^I | } \sum_{ q \in c_{i-1}(p)\cap Q^I } \dist^2(q,Z)
	\end{align}
	where the first step follows from triangle inequality, the second step follows from definition of the grids, the last step follows from an averaging argument.

	According to Claim~\ref{cla:heavy_cell_contains_important_points}, we have
	\begin{align*}
	| c_{i-1}(p)\cap Q^I | \geq \frac{1}{2} T_{i-1}(o).
	\end{align*}
	
	Let $Q^N = Q\setminus Q^I$. 
	We can lower bound $\cost(Q,Z)$ in the following sense,
	\begin{align*}
	\cost(Q,Z) = & ~ \cost(Q^I, Z) + \cost(Q^N, Z) \\
	= & ~ \cost(Q^I, Z) + \sum_{i\not\in I}\sum_{p\in Q_i} \dist^2(p,z) \\
	\leq & ~ \cost(Q^I, Z) + 2 \sum_{i \not\in I} \sum_{p\in Q_i}  \left( d g_{i-1}^2 + \frac{ 1 }{ | c_{i-1}(p)\cap Q^I | } \sum_{ q \in c_{i-1}(p)\cap Q^I } \dist^2(q,Z) \right) \\
	\leq & ~ \cost(Q^I, Z) + 2 \sum_{i \not\in I} \sum_{p\in Q_i}  \left( d g_{i-1}^2 + \frac{ 2 }{ T_{i-1}(o) } \sum_{ q \in c_{i-1}(p)\cap Q^I } \dist^2(q,Z) \right) \\
	\leq & ~ \cost(Q^I, Z) + 2 \sum_{i \not\in I} \sum_{p \in Q_i} \left( d g_{i-1}^2 + \frac{ 2 }{ T_{i-1}(o) } \sum_{ q \in Q^I } \dist^2(q,Z) \right) \\
	\leq & ~ \cost(Q^I, Z) + 2 \sum_{i \not\in I} 1.1\gamma T_i(o) \cdot \left( d g_{i-1}^2 + \frac{ 2 }{ T_{i-1}(o) } \sum_{ q \in Q^I } \dist^2(q,Z) \right) \\
	\leq & ~ \cost(Q^I, Z) + 5L \gamma T_i(o) \cdot \left( d g_{i-1}^2 + \frac{ 2 }{ T_{i-1}(o) } \sum_{ q \in Q^I } \dist^2(q,Z) \right) \\
	= & ~ \cost(Q^I, Z) + 5L \gamma T_i(o) \cdot \left( d g_{i-1}^2 + \frac{ 2 }{ T_{i-1}(o) } \cost( Q^I , Z ) \right) \\
	\leq & ~ \cost(Q^I, Z) + 5 L \gamma ( d^3  o / (25k) + 8 \cost (Q^I, Z) ) \\
	\leq & ~ \cost(Q^I, Z) + 5 L \gamma ( d^3  \cost(Q,Z) / (25k) + 8 \cost (Q^I, Z) ).
	\end{align*}
	where the second step follows from the definition of the cost, the third step follows from Eq.~\eqref{eq:averaging_argument}, the fourth step follows from $|c_{i-1}(p)\cap Q^I| \geq T_{i-1}(o) /2 $, the fifth step follows from $(c_{i-1}(p)\cap Q^I) \subset Q^I$, the sixth step follows from $|Q_i| \leq 1.1\gamma T_i(o)$, the seventh step follows from $L+1-|I|\leq L+1 \leq 2L$ and $2\cdot 2\cdot 1.1\leq 5$, the ninth step follows from $T_i(o) = 4 T_{i-1}(o)$ and $T_i(o) = (d/g_i)^2 \cdot  o / k\cdot 1/100$, and the last step follows from $o \leq \OPT \leq \cost(Q,Z)$.
	
	It implies that
	\begin{align*}
	\frac{ \cost( Q , Z ) }{ \cost( Q^I , Z ) } \leq \frac{1+ 40 L \gamma}{1-5 L\gamma d^3 / (25k)} \leq \frac{1+\epsilon/40}{1-\epsilon/40} \leq 1+\epsilon/10,
	\end{align*}
	where the second step follows from $\gamma \leq \epsilon / (40^2 L d^3)$, and the last step follows from $\varepsilon < 1$.
\end{proofof}

\begin{proofof}{Lemma~\ref{lem:sensitivity_ub}}
Let $Z\subseteq \mathbb{R}^d$ be an arbitrary set of $k$ centers. Fix a point $p\in Q$. Suppose $p$ is in $Q_i$, i.e., $p$ is in a crucial cell of $G_i$. Let $C=c_{i-1}(p)$, i.e., $C$ is the parent cell of the crucial cell that contains $p$. By Algorithm~\ref{alg:offline_sens}, $C$ and all of its ancestors must be heavy. By Claim~\ref{cla:heavy_cell_contains_important_points}, $C\cap Q^I$ cannot be empty. Thus, by an averaging argument, there is a point $p'\in C$ such that 
\begin{align}\label{eq:average_arugment_in_a_cell}
\dist^2(p',Z)\leq \frac{1}{|C \cap Q^I|} \sum_{q\in C\cap Q^I} \dist^2(q, Z).
\end{align}

We have
\begin{align*}
\frac{\dist^2(p,Z)}{ \sum_{q \in Q^I} \dist^2 (q,Z) }
\leq & ~ 2 \frac{\dist^2(p',Z)}{ \sum_{q \in Q^I} \dist^2 (q,Z) } + 2 \frac{\dist^2(p,p')}{ \sum_{q \in Q^I} \dist^2 (q,Z) }  \\
\leq & ~ 2 \frac{ \sum_{q\in C \cap Q^I } \dist^2(q,Z)  }{  |C \cap Q^I| \sum_{q \in Q^I} \dist^2 (q,Z) }  + 2 \frac{ d g_{i-1}^2 }{ \sum_{q\in Q^I} \dist^2(q,Z)} \\
\leq & ~ 2 \frac{ \sum_{q\in Q^I } \dist^2(q,Z)  }{  |C \cap Q^I| \sum_{q \in Q^I} \dist^2 (q,Z) }  + 2 \frac{ d g_{i-1}^2 }{ \sum_{q\in Q^I} \dist^2(q,Z)} \\
= & ~ 2 \frac{1}{|C \cap Q^I|} + 2 \frac{ d g_{i-1}^2 }{ \sum_{q\in Q^I} \dist^2(q,Z)} \\
\leq & ~ 2 \frac{1}{|C \cap Q^I|} + 4 \frac{ d g_{i-1}^2 }{ \OPT } \\
\leq & ~ 2 \frac{1}{|C \cap Q^I|} + 16 \frac{ d  g_{i}^2 }{ \OPT } \\
\leq & ~ 2 \frac{1}{|C \cap Q^I|} +  \frac{ d^3  o  }{ T_i(o) k \OPT } \\
\leq & ~ 9 \frac{1}{ T_i(o)} +  \frac{ d^3  o  }{ T_i(o) k \OPT } \\
\leq & ~ 9 \frac{1}{T_i(o)} +  \frac{d^3}{T_i(o)} \\
\leq & ~ 10 \frac{d^3}{T_i(o)}
\end{align*}
where the first step follows from triangle inequality, the second step follows from Eq.~\eqref{eq:average_arugment_in_a_cell} and $p' \in c_{i-1}( p )$,  the fifth step follows from $\sum_{q\in Q^I } \dist^2(q,Z) \geq (1-\epsilon) \cost(Q,Z) \geq \OPT /2$ (Lemma~\ref{lem:reducing_to_important_points}), the sixth step follows from $g_{i-1}^2 \leq 4 g_i^2$, the seventh step follows from $T_i(o) = \frac{d^2}{g_i^2} \frac{ o}{ 100 k}$, the eighth step follows from $T_i(o) = 4 T_{i-1}(o) \leq 4.5 |C\cap Q^I|$ (Claim~\ref{cla:heavy_cell_contains_important_points} and $|C\cap Q|\geq 0.9 T_{i-1}(o)$), 
the ninth step follows from $ 1/k \leq 1, o\leq \OPT$.
\end{proofof}

\begin{proofof}{Lemma~\ref{lem:num_center_cell}}
	Fix an $i\in\{0,1,\ldots,L\}$ and consider the grid $G_i$. For each optimal center $z_j^*$, we use $X_{j,\alpha}$ to denote the indicator random variable for the event that the distance from $z_j^*$ to the boundary in dimension $\alpha$ of the grid $G_i$ is at most $g_i / (2d)$. Since in each dimension, if the center is close to a boundary, it contributes a factor at most $2$ to the total number of center cells. It follows that the number of cells that have distance at most $g_i/(2d)$ to $z_j^*$ is at most
	\begin{align*}
	N = 2^{\sum_{\alpha=1}^d X_{j,\alpha}}.
	\end{align*}
	We denote $Y_{j,\alpha}$ to be $ 2^{X_{j,\alpha}}$, then
	\begin{align*}
	\E[ N ] = \E \left[ \prod_{\alpha=1}^d Y_{j,\alpha}  \right] = \prod_{\alpha=1}^d \E[ Y_{j,\alpha} ].
	\end{align*}
	By using $\Pr[ X_{j,\alpha} =1] \leq (2g_i/(2d))/g_i =1/d $, we obtain
	\begin{align*}
	\E[ Y_{j,\alpha} ] \leq \E[ 1 + X_{j,\alpha} ] = 1 + \E[ X_{j,\alpha} ] \leq 1 + 1/d.
	\end{align*}
	Thus $\E[ N ] = \prod_{\alpha=1}^d \E[Y_{j,\alpha}] \leq (1+1/d)^d \leq e$. The expected number of center cells in a single grid is at most $(1+1/d)^{d} k \leq e k\leq 3k$.
	By linearity of expectation, the expected number of center cells in all grids is at most $ek(L+1)\leq 6kL$.
	By Markov's inequality, the probability that we have more than $100 kL$ center cells in all grids is at most $0.06$.
\end{proofof}

\begin{proofof}{Lemma~\ref{lem:total_sensitivity_ub}}
\begin{align*}
\sum_{p\in Q} s'(p) & = \sum_{i=0}^L \sum_{p\in Q_i} 10\frac{d^3}{T_i(o)} \\
& = 10d^3 \sum_{i=0}^L \frac{1}{T_i(o)} \left( \sum_{\text{center cell } C\in G_i} |C\cap Q_i| + \sum_{\text{non-center cell } C\in G_i} |C\cap Q_i| \right)\\
&\leq 10d^3 \sum_{i=0}^L  \frac{1}{T_i(o)}  \left( \sum_{\text{center cell } C\in G_i} 1.1 T_i(o) + \frac{\OPT}{g_i^2/(2d)^2} \right)\\
&\leq 11d^3 \cdot (\text{\# of center cells}) + 4000d^3 L k  \cdot \frac{\OPT}{o}\\
&\leq 1100d^3Lk + 4000d^3 L k  \cdot \frac{\OPT}{o}\\
&\leq 4000d^3 L k  \cdot \left(\frac{\OPT}{o} + 1\right),
\end{align*}
where the first step follows from the definition of $s'(p)$, the third step follows from 
\begin{enumerate}
	\item if $C\in G_i$ and $C\cap Q_i\not=\emptyset$, then $|C\cap Q_i|=|C\cap Q|\leq 1.1 T_i(o)$;
	\item if $p$ is in a non-center cell $C\in G_i$, then $\dist^2(p,Z^*)\geq g_i^2/(2d)^2$,
\end{enumerate}
the forth step follows from $g_i^2/d^2 = o/(100kT_i(o))$, the fifth step follows from that the total number of center cells is bounded by $100kL$.
\end{proofof}

\section{Missing Details in Section~\ref{sec:streaming}}\label{sec:missing_details_of_dynamic_stream}

In this section, we give all the missing details in Section~\ref{sec:streaming}.

\subsection{The Dynamic Point-Cell Storing Data Structure}

The following lemma shows the guarantee of Algorithm~\ref{alg:storing}.

\begin{lemma}\label{lem:cell_point_storing}
	Given parameters $i\in[L],\alpha,\beta\in\mathbb{N}_+,\delta\in(0,0.5)$, \textsc{Storing}$(G_i,\alpha,\beta,\delta)$ (Algorithm~\ref{alg:storing}) uses $O(\alpha\beta dL\cdot\log^2(\alpha\beta/\delta))$ bits to process a stream $\{((p_1,l_1),\pm),((p_2,l_2),\pm),\ldots\}$ of insertion/deletion operations of data points. At the end of the stream, if the number of non-empty cells in $G_i$ is at most $\alpha$, then with probability at least $1-\delta$ it returns the set $\mathcal{C}$ of all the non-empty cells, the number of points $f(C)$ in each cell $C\in \mathcal{C}$, and the set $S$ of points in all the non-empty cells that contain at most $\beta$ points. It returns \textbf{FAIL} otherwise.
\end{lemma}
\begin{proof}
	If the number of non-empty cells of $G_i$ at the end of the stream is more than $\alpha$, then according to line~\ref{sta:global_estimation} of Algorithm~\ref{alg:storing} and Lemma~\ref{lem:k_set}, Algorithm~\ref{alg:storing} must output \textbf{FAIL}.
	
	Now consider the case when the total number of non-empty cells is at most $\alpha$. According to Lemma~\ref{lem:k_set}, with probability at least $1-\delta/4$, line~\ref{sta:global_estimation} of Algorithm~\ref{alg:storing} will return the set $\mathcal{C}$ of all the non-empty cells of $G_i$ and the number of points $f(C)$ for each cell $C\in \mathcal{C}$. For each $C\in\mathcal{C}$ with $|C|\leq \beta$, since $|\mathcal{C}|\leq \alpha,$ the probability that $\exists j\in[r]$ such that $\forall C'\in \mathcal{C},h_j(C')\not=h_j(C)$ is at least $1-1/2^r\geq 1-\delta/(4\alpha)$ and furthermore the probability that the $(j,h_j(C))$-th copy of \textsc{Distinct}$(\beta,\delta/(2\alpha))$ in line~\ref{sta:multi_copy_distinct} of Algorithm~\ref{alg:storing} will output all the points in $C$ is at least $1-\delta/(2\alpha)$. By taking union bound, the overall probability that Algorithm~\ref{alg:storing} does not output \textbf{FAIL} is at most $1-\delta/4-(\delta/(2\alpha)+\delta/(4\alpha))\cdot \alpha=1-\delta$.
	
	According to Lemma~\ref{lem:k_set}, the space needed by line~\ref{sta:global_estimation} of Algorithm~\ref{alg:storing} is $O(\alpha dL\cdot \log(\alpha/\delta))$ bits and the space needed of each copy in line~\ref{sta:multi_copy_distinct} of Algorithm~\ref{alg:storing} is $O(\beta dL \cdot \log(\alpha\beta/\delta))$ bits. Thus, the total space needed is at most $O(\alpha\beta dL\cdot\log^2(\alpha\beta/\delta))$.
\end{proof}

\subsection{Estimating the Number of Points in Each Cell}

Now let us analyze Algorithm~\ref{alg:streaming_est_cell}.
We need the following high concentration bound in our analysis.
\begin{theorem}[\cite{br94}]
	\label{thm:lambda_wise_concentration}
	Let $\lambda$ be an even integer, and let $X$ be the sum of $n$ $\lambda$-wise independent random variables taking values in $[0,1]$. Let $\mu=\E[X]$ and $a>0$. Then we have
	\[
	\Pr\bigg[|X-\mu|>a\bigg]\le8\cdot\bigg(\frac{\lambda\mu+\lambda^2}{a^2}\bigg)^{\lambda/2}.
	\]
\end{theorem}

\begin{lemma}[Samples from each cell]\label{lem:good_samples}
	In \textsc{PointsEstimation}$(o,\varepsilon,\delta)$ (Algorithm~\ref{alg:streaming_est_cell}), with probability at least $1-\delta/10$, $\forall i \in \{0,1,\ldots,L\}$ with $T_i(o)\geq 4\cdot 10^4\lambda$, $\forall C\in G_i$, we have either 
	\begin{align*}
	\sum_{p \in C\cap Q} h_i(p) \in |C\cap Q|\cdot \frac{4\cdot 10^4\lambda}{T_i(o)} \pm 4\cdot 10^3\lambda
	\end{align*}
	or
	\begin{align*}
	\sum_{p \in C\cap Q} h_i(p) \in |C\cap Q| \cdot \frac{4\cdot 10^4\lambda}{T_i(o)} \cdot(1\pm 0.01).
	\end{align*}
\end{lemma}
\begin{proof}
	Consider a cell $C\in G_i$.
	If $|C\cap Q|\leq T_i(o)$, then $\mu=\E\left[\sum_{p\in C\cap Q} h_i(p)\right]=4\cdot 10^4\lambda|C\cap Q|/T_i(o)\leq 4\cdot 10^4 \lambda.$
	According to Theorem~\ref{thm:lambda_wise_concentration}, we have
	\begin{align*}
	\Pr\left[\left|\sum_{p\in C\cap Q} h_i(p) - \mu\right| > 4\cdot 10^3 \lambda \right]\leq 8\cdot \left(\frac{4\cdot 10^4 \lambda^2+\lambda^2}{(4\cdot 10^3\lambda)^2}\right)^{\lambda/2}\leq 8\cdot (1/2)^{\lambda/2}\leq (\delta/\Delta^d)^{5}.
	\end{align*}
	If $|C\cap Q|> T_i(o)$, then $\mu=\E\left[\sum_{p\in C\cap Q} h_i(p)\right]=4\cdot 10^4\lambda|C\cap Q|/T_i(o)> 4\cdot 10^4\lambda.$ According to Theorem~\ref{thm:lambda_wise_concentration}, we have
	\begin{align*}
	\Pr\left[\left|\sum_{p\in C\cap Q} h_i(p) - \mu\right| > 0.01\mu \right]\leq 8\cdot \left(\frac{\lambda \mu+\lambda^2}{(0.01\mu)^2}\right)^{\lambda/2}\leq 8\cdot (1/2)^{\lambda/2}\leq (\delta/\Delta^d)^{5}.
	\end{align*}
	By taking union bound over all $i\in\{0,1,\ldots,L\}$ and all the cells in $G_i$, the claim is proved.
\end{proof}

\begin{lemma}[Estimating the number of points in each cell]\label{lem:good_estimation_for_each_cell}
	If \textsc{PointsEstimation}$(o,\varepsilon,\delta)$ (Algorithm~\ref{alg:streaming_est_cell}) does not output \textbf{FAIL}, then with probability at least $1-\delta/10$, $\forall i\in \{0,1,\ldots,L\}, C\in G_i$, either $\hat{f}(C)\in |C\cap Q|\pm 0.1T_i(o)$ or $\hat{f}(C)\in (1\pm 0.01) |C\cap Q|$.
\end{lemma}
\begin{proof}
	Suppose \textsc{PointsEstimation}$(o,\varepsilon,\delta)$ does not output \textbf{FAIL}. According to Lemma~\ref{lem:cell_point_storing}, $\forall i\in\{0,1,\ldots,L\}$ and $\forall C\in G_i$, $f_i(C)=\sum_{p\in C\cap Q}h_i(p)$. If $T_i(o)\leq 4\cdot 10^4 \lambda$, then $\hat{f}(C)=f_i(C)=\sum_{p\in C\cap Q}h_i(p) = |C\cap Q|$.
	
	Due to Lemma~\ref{lem:good_samples}, with probability at least $1-\delta/10$, $\forall i\in\{0,1,\ldots,L\}$ with $T_i(o)>4\cdot 10^4 \lambda$ and $\forall C\in G_i$, either $f_i(C)\in |C\cap Q|\cdot \frac{4\cdot 10^4\lambda}{T_i(o)}\pm 4\cdot 10^3\lambda$ or $f_i(C)\in |C\cap Q|\cdot \frac{4\cdot 10^4\lambda}{T_i(o)}\cdot (1\pm 0.01)$. Thus, either $\hat{f}(C)\in |C\cap Q|\pm 0.1T_i(o)$ or $\hat{f}(C)\in (1\pm0.01)|C\cap Q|.$
\end{proof}

\begin{lemma}[Samples from each $Q_i$]\label{lem:good_qi} In \textsc{PointsEstimation}$(o,\varepsilon,\delta)$ (Algorithm~\ref{alg:streaming_est_cell}), with probability at least $1-\delta/10$, $\forall i \in \{0,1,\ldots, L\}$ with $T_i(o)\geq 4\cdot 10^4\varepsilon^{-2}\gamma^{-1}\lambda$, $\forall C\in G_i$, either
	\begin{align*}
	\sum_{p\in Q_i} h'_i(p) \in |Q_i|\cdot \frac{4\cdot 10^4\varepsilon^{-2}\gamma^{-1}\lambda}{T_i(o)} \pm 4\cdot 10^3 \varepsilon^{-1} \lambda
	\end{align*}
	or
	\begin{align*}
	\sum_{p\in Q_i} h'_i(p) \in |Q_i|\cdot \frac{4\cdot 10^4\varepsilon^{-2}\gamma^{-1}\lambda}{T_i(o)} \cdot (1\pm 0.01\varepsilon).
	\end{align*}
\end{lemma}

\begin{proof}
	If $|Q_i|\leq \gamma T_i(o)$, then $\mu=\E\left[\sum_{p\in Q_i} h_i'(p)\right]=4\cdot 10^4\varepsilon^{-2}\gamma^{-1}\lambda| Q_i|/T_i(o)\leq 4\cdot 10^4 \varepsilon^{-2}\lambda.$
	According to Theorem~\ref{thm:lambda_wise_concentration}, we have
	\begin{align*}
	\Pr\left[\left|\sum_{p\in Q_i} h_i'(p) - \mu\right| > 4\cdot 10^3 \varepsilon^{-1}\lambda \right]\leq 8\cdot \left(\frac{4\cdot 10^4 \varepsilon^{-2} \lambda^2+\lambda^2}{(4\cdot 10^3\varepsilon^{-1}\lambda)^2}\right)^{\lambda/2}\leq 8\cdot (1/2)^{\lambda/2}\leq (\delta/\Delta^d)^{5}.
	\end{align*}
	If $|Q_i|> \gamma T_i(o)$, then $\mu=\E\left[\sum_{p\in Q_i} h_i'(p)\right]=4\cdot 10^4\varepsilon^{-2}\gamma^{-1}\lambda|Q_i|/T_i(o)> 4\cdot 10^4\varepsilon^{-2}\lambda.$ According to Theorem~\ref{thm:lambda_wise_concentration}, we have
	\begin{align*}
	\Pr\left[\left|\sum_{p\in Q_i} h_i'(p) - \mu\right| > 0.01\varepsilon\mu \right]\leq 8\cdot \left(\frac{\lambda \mu+\lambda^2}{(0.01\varepsilon\mu)^2}\right)^{\lambda/2}\leq 8\cdot (1/2)^{\lambda/2}\leq (\delta/\Delta^d)^{5}.
	\end{align*}
	By taking union bound over all $i\in\{0,1,\ldots,L\}$, the claim is proved.
\end{proof}

\begin{lemma}[$\hat{q}_i$ can estimate $|Q_i|$ well]\label{lem:qi_esimate_well}
	If \textsc{PointsEsitmation}$(o,\varepsilon,\delta)$ (Algorithm~\ref{alg:streaming_est_cell}) does not output \textbf{FAIL}, then with probability at least $1-\delta/10$, $\forall i\in\{0,1,\ldots,L\}$, either $\hat{q}_i\in |Q_i|\pm 0.1\varepsilon \gamma T_i(o)$ or $\hat{q}_i \in (1\pm 0.01\varepsilon)|Q_i|$.
\end{lemma}

\begin{proof}
	Suppose \textsc{PointsEstimation}$(o,\varepsilon,\delta)$ does not output \textbf{FAIL}. According to Lemma~\ref{lem:cell_point_storing}, $\forall i\in\{0,1,\ldots,L\}$, $f_i'(C)=\sum_{p\in Q_i}h_i'(p)$. If $T_i(o)\leq 4\cdot 10^4 \varepsilon^{-2}\gamma^{-1}\lambda$, then $\hat{q}_i=\sum_{C\in G_i:C\text{ is crucial}}f_i'(C)=\sum_{p\in Q_i}h_i'(p) = |Q_i|$.
	
	Due to Lemma~\ref{lem:good_qi}, with probability at least $1-\delta/10$, $\forall i\in\{0,1,\ldots,L\}$ with $T_i(o)>4\cdot 10^4 \varepsilon^{-2}\gamma^{-1} \lambda$, either $$\sum_{C\in G_i:C\text{ is crucial}} f_i'(C)\in |Q_i|\cdot \frac{4\cdot 10^4\varepsilon^{-2}\gamma^{-1}\lambda}{T_i(o)}\pm 4\cdot 10^3\varepsilon^{-1}\lambda$$ or $$\sum_{C\in G_i:C\text{ is crucial}} f_i'(C)\in |Q_i|\cdot \frac{4\cdot 10^4\varepsilon^{-2}\gamma^{-1}\lambda}{T_i(o)}\cdot (1\pm 0.01\varepsilon).$$ Thus, either $\hat{q}_i\in |Q_i|\pm 0.1\varepsilon \gamma T_i(o)$ or $\hat{q}_i\in (1\pm0.01\varepsilon)|Q_i|$.
\end{proof}

\begin{lemma}[Number of points sampled from non-center cells]\label{lem:samples_from_non_center_cell}
	In \textsc{PointsEsitmation}$(o,\varepsilon,\delta)$ (Algorithm~\ref{alg:streaming_est_cell}), with probability at least $1-\delta/10$, $\forall i\in \{0,1,\ldots,L\}$, we have
	\begin{align*}
	\sum_{\text{non-center cell }C\in G_i}\sum_{p\in C \cap Q} h_i(p)\leq 4\cdot 10^3 \lambda + 1.01\cdot \frac{4\cdot 10^4\lambda}{T_i(o)}\sum_{\text{non-center cell }C\in G_i}|C\cap Q|
	\end{align*}
	and
	\begin{align*}
	\sum_{\text{non-center cell }C\in G_i}\sum_{p\in C \cap Q} h_i'(p)\leq 4\cdot 10^3 \varepsilon^{-1}\lambda + 1.01\cdot \frac{4\cdot 10^4\varepsilon^{-2}\gamma^{-1}\lambda}{T_i(o)}\sum_{\text{non-center cell }C\in G_i}|C\cap Q|.
	\end{align*}
\end{lemma}
\begin{proof}
	The proof is exactly the same as the proofs of Lemmas~\ref{lem:good_samples} and \ref{lem:good_qi}.
\end{proof}

\begin{lemma}\label{lem:points_in_non_center_cell}
	$\forall i\in\{0,1,\ldots,L\}$, 
	\begin{align*}
	\sum_{\text{non-center cell }C\in G_i} \frac{|C\cap Q|}{T_i(o)} \leq 400k\cdot (\OPT/o).
	\end{align*}
\end{lemma}
\begin{proof}
	\begin{align*}
	&\sum_{\text{non-center cell }C\in G_i} |C\cap Q|\\
	\leq & \frac{\OPT}{(g_i/(2d))^2} = 400k T_i(o) \cdot \frac{\OPT}{o}.
	\end{align*}
\end{proof}

\begin{lemma}[The success probability]\label{lem:point_estimation_suc_prob}
	Condition on the number of center cells of all the grids is at most $100kL$.
	If $o\in (\OPT/16,\OPT]$, then \textsc{PointsEstimation}$(o,\varepsilon,\delta)$ (Algorithm~\ref{alg:streaming_est_cell}) does not output \textbf{FAIL} with probability at least $1-3\delta/10$.
\end{lemma}
\begin{proof}
	Let $o\geq \OPT/16$, according to Lemma~\ref{lem:points_in_non_center_cell} and Lemma~\ref{lem:samples_from_non_center_cell}, with probability at least $1-\delta/10,$ $\forall i\in \{0,1,\ldots, L\}$, we have 
	\begin{align*}
	\sum_{\text{non-center cell }C\in G_i}\sum_{p\in C \cap Q} h_i(p)\leq 4\cdot 10^3 \lambda + 1.01\cdot 4\cdot 10^4\lambda\cdot 6400k\leq 10^{11} kdL\log(1/\delta)
	\end{align*} 
	and
	\begin{align*}
	\sum_{\text{non-center cell }C\in G_i}\sum_{p\in C \cap Q} h_i'(p)\leq 4\cdot 10^3 \varepsilon^{-1}\lambda + 1.01\cdot 4\cdot 10^4\varepsilon^{-2}\gamma^{-1}\lambda\cdot 6400 k\leq 10^{15}\cdot \varepsilon^{-3}kL^2d^4.
	\end{align*}
	Since the number of center cells is at most $100kL$, the total number of cells which contains some $p$ with $h_i(p)=1$ is at most $10^{12}kdL\log(1/\delta)\leq \alpha$, and the total number of cells which contains some $p$ with $h_i'(p)=1$ is at most $10^{16}\varepsilon^{-3}kL^2d^4\leq \alpha'$.
	According to Lemma~\ref{lem:cell_point_storing}, with probability at least $1-2\delta/10$, none the call of \textsc{Storing} will return \textbf{FAIL}. Thus the overall probability that Algorithm~\ref{alg:streaming_est_cell} does not output \textbf{FAIL} is at last $1-3\delta/10$.
\end{proof}

\begin{lemma}[Space of Algorithm~\ref{alg:streaming_est_cell}]\label{lem:space_points_estimation}
	\textsc{PointsEstimation}$(o,\varepsilon,\delta)$ (Algorithm~\ref{alg:streaming_est_cell}) uses space at most $O(k\varepsilon^{-3}L^4 d^5\cdot \log^2(kLd/(\varepsilon\delta)))$ bits.
\end{lemma}
\begin{proof}
	The total space used is dominated by the space needed to run $L+1$ copies of \textsc{Storing}$(G_i,\alpha',1,0.1\delta/L)$ in line~\ref{sta:call_storing_dominating} of Algorithm~\ref{alg:streaming_est_cell}. According to Lemma~\ref{lem:cell_point_storing}, the total space needed is $O(L\cdot \alpha' dL\log^2(L\alpha'/\delta))=O(k\varepsilon^{-3}L^4 d^5\cdot \log^2(kLd/(\varepsilon\delta)))$ bits.
\end{proof}

\subsection{Sensitivity Sampling over a Dynamic Stream}

Now we analyze Algorithm~\ref{alg:sensitivity_sampler}.

\begin{fact}\label{fac:sampling_implementation}
	If \textsc{Sampling}$(o,\varepsilon,\delta)$ (Algorithm~\ref{alg:sensitivity_sampler}) does not output \textbf{FAIL}, then line~\ref{sta:uniform_sampling_doable} can be implemented, and $p$ is a uniform sample drawn from $Q_i$.
\end{fact}
\begin{proof}
	Although $Q_i$ cannot be stored explicitly, $\forall p\in Q,$ we are able to determine whether $p\in Q_i$ since we can use $\hat{f}$ to find all the crucial cells and check whether $p$ is in a crucial cell of $G_i$.
	
	Suppose \textsc{Sampling}$(o,\varepsilon,\delta)$ does not output \textbf{FAIL}.
	According to Lemma~\ref{lem:cell_point_storing} and the condition in line~\ref{sta:guarantee_storing_correctly}, $\forall j\in[\hat{m}],$ we have $\{(p,j)\mid p\in Q_i,h_{i,j}(p)=1\}\subseteq S_i$. Since $\forall x,y\in Q_i$, $\Pr[h_{i,j}(x)=1]=\Pr[h_{i,j}(y)=1]$, then the sample $p$ in line~\ref{sta:uniform_sampling_doable} is drawn uniformly from $Q_i$.
\end{proof}

\begin{lemma}[Correctness of Algorithm~\ref{alg:sensitivity_sampler}]\label{lem:correctness_of_sampler}
	Suppose $o\in(0,\OPT]$.
	If \textsc{Sampling}$(o,\varepsilon,\delta)$ (Algorithm~\ref{alg:sensitivity_sampler}) does not output \textbf{FAIL}, then with probability at least $1-\delta/5$, the output $S$ by \textsc{Sampling}$(o,\varepsilon,\delta)$ is an $\varepsilon$-coreset for $Q$. Furthermore, the size $|S|$ is at most $O(k\varepsilon^{-2}L^2d^4\log(kLd)\cdot (\OPT/o + 1))$.
\end{lemma}
\begin{proof}
	Due to Lemma~\ref{lem:good_estimation_for_each_cell} and Lemma~\ref{lem:qi_esimate_well}, with probability at least $1-\delta/10$, $\forall C\in\bigcup_{i=0}^L G_i$, either $\hat{f}(C)\in |C\cap Q|\pm 0.1 T_i(o)$ or $\hat{f}(C)\in (1\pm 0.01)|C\cap Q|$ and $\forall i\in\{0,1,\ldots,L\}$, either $\hat{q}_i\in |Q_i|\pm 0.1\varepsilon\gamma T_i(o)$ or $\hat{q}_i \in (1\pm 0.01\varepsilon)|Q_i|$, where $Q_0,Q_1,\ldots,Q_L$ are defined by using the estimation $\hat{f}(\cdot)$ (Algorithm~\ref{alg:offline_sens} or Algorithm~\ref{alg:streaming_est_cell}). According to Fact~\ref{fac:sampling_implementation}, the sampling procedure can be implemented. Then by Theorem\footnote{The proof is slightly different since Theorem~\ref{thm:offline_sampling_coreset} only claims a constant success probability. See Section~\ref{sec:offline_analysis} for the detailed proof of Theorem~\ref{thm:offline_sampling_coreset}.}~\ref{thm:offline_sampling_coreset}, with probability at least $1-\delta/10$, the output set $S$ is an $\varepsilon$-coreset for $Q$. By taking union bound, with probability $1-\delta/5$, the set $S$ is an $\varepsilon$-coreset for $Q$. 
\end{proof}

Now let us consider the success probability of Algorithm~\ref{alg:sensitivity_sampler}. Since we know the success probability of \textsc{PointsEstimation}$(o,\varepsilon,\delta/2)$ in line~\ref{sta:run_dynamic_partition} of Algortihm~\ref{alg:sensitivity_sampler} and the success probability of \textsc{Storing}$(G_i,\alpha,\beta,0.1\delta/L)$ in line~\ref{sta:store_samples}, we only need to analyze the success probability in line~\ref{sta:find_a_sample} of Algorithm~\ref{alg:sensitivity_sampler}. To make line~\ref{sta:find_a_sample} succeed, we need to find enough samples from $Q_i$, i.e., we hope that $\sum_{j=1}^{\hat{m}} \mathbf{1}(|\{p\in Q_i\mid h_{i,j}(p)=1\}|>0)$ is large.
In the following analysis, we will show that $\sum_{j=1}^{\hat{m}} \mathbf{1}(|\{p\in Q_i\mid h_{i,j}(p)=1\}|>0)$ is large. First, we show that the number of samples drawn from level $i$ is bounded.

\begin{lemma}[Number of samples from each level]\label{lem:num_samples_requested_each_level}
	Let $I$ be the set computed in \textsc{Sampling}$(o,\varepsilon,\delta)$ (Algortihm~\ref{alg:sensitivity_sampler}).
	With probability at least $1-\delta/10$, $\forall i\in I$, the number of times that $i$ is chosen in line~\ref{sta:choose_i} of Algorithm~\ref{alg:sensitivity_sampler} is at most 
	\begin{align*}
	O\left(\varepsilon^{-2}Ld^4\log (t'/\delta)\cdot\frac{L}{\delta}\cdot\frac{\hat{q}_i}{T_i(o)}\right).
	\end{align*}
\end{lemma}
\begin{proof}
	For $i\in I$, the expected number of times that $i$ is chosen is $O(m\cdot \hat{q}_i\cdot d^3/T_i(o)/t')$. By Markov's inequality, with probability at least $1-\delta/(20L)$, the number of times that $i$ is chosen in line~\ref{sta:choose_i} of Algorithm~\ref{alg:sensitivity_sampler} is at most $O\left(\varepsilon^{-2}Ld^4\log (t'/\delta)\cdot\frac{L}{\delta}\cdot\frac{\hat{q}_i}{T_i(o)}\right).$ By taking union bound over all $i\in I$, we complete the proof.
\end{proof}

\begin{lemma}[Bounding $t'$]\label{lem:bound_of_tprime}
	Consider $o\geq \OPT/16$. Condition on $\hat{f}:\bigcup_{i=0}^L G_i\rightarrow \mathbb{R}_+$ (in Algorithm~\ref{alg:sensitivity_sampler}) is good (Lemma~\ref{lem:good_samples}), $\hat{q}_0,\hat{q}_1,\ldots,\hat{q}_L$ (in Algorithm~\ref{alg:sensitivity_sampler}) are good (Lemma~\ref{lem:good_qi}), and the number of center cells is at most $100kL$, then we have $t'\leq 10^6 d^3Lk$.
\end{lemma}
\begin{proof}
	\begin{align*}
	t' & = \sum_{i\in I} \hat{q}_i\cdot 10 d^3/T_i(o)\\
	& \leq \sum_{i\in I} |Q_i|\cdot 20 d^3/T_i(o)\\
	& \leq \sum_{i=0}^L |Q_i|\cdot 20 d^3/T_i(o)\\
	& \leq 2\cdot 4000d^3Lk\cdot 20\\
	& \leq 10^6 d^3 Lk,
	\end{align*}
	where the first inequality follows by $\forall i\in I,\hat{q}_i\geq \gamma T_i(o)$ and either $\hat{q}_i\in |Q_i|\pm 0.1\varepsilon\gamma T_i(o)$ or $\hat{q}_i\in (1\pm 0.01\varepsilon)|Q_i|$, the second inequality follows by $I\subseteq \{0,1,\ldots,L\}$, the third inequality follows by Lemma~\ref{lem:total_sensitivity_ub} and $o\geq \OPT/16$. 
\end{proof}

\begin{lemma}[Number of crucial points in each level]\label{lem:crucial_points_in_each_level}
	Consider $o\geq \OPT/16$. Conditioning on $\hat{f}:\bigcup_{i=0}^L G_i\rightarrow \mathbb{R}_+$ (in Algorithm~\ref{alg:sensitivity_sampler}) is good (Lemma~\ref{lem:good_estimation_for_each_cell}) and the number of center cells is at most $100kL$, we have:
	\begin{align*}
	\frac{|Q_i|}{T_i(o)}\leq 7000kL.
	\end{align*}
\end{lemma}
\begin{proof}
	\begin{align*}
	\frac{|Q_i|}{T_i(o)}&= \sum_{\text{center crucial cell }C\in G_i} \frac{|C\cap Q|}{T_i(o)} + \sum_{\text{non-center crucial cell }C\in G_i} \frac{|C\cap Q|}{T_i(o)}\\
	&\leq 110kL + 400k(\OPT/o)\\
	&\leq 7000kL,
	\end{align*}
	where the first inequality follows by that the number of center cells is at most $100kL$, the number points in a crucial cell is at most $1.1T_i(o)$ and Lemma~\ref{lem:points_in_non_center_cell}.
\end{proof}

\begin{lemma}[The number of samples is large]\label{lem:num_samples_in_each_level_is_enough}
	Consider $o\geq \OPT/16$. Conditioning on $\hat{f}:\bigcup_{i=0}^L G_i\rightarrow \mathbb{R}_+$ (in Algorithm~\ref{alg:sensitivity_sampler}) is good (Lemma~\ref{lem:good_samples}), $\hat{q}_0,\hat{q}_1,\ldots,\hat{q}_L$ (in Algorithm~\ref{alg:sensitivity_sampler}) are good (Lemma~\ref{lem:good_qi}), and the number of center cells is at most $100kL$, with probability at least $1-\delta/10$, $\forall i\in I$, we have:
	\begin{align*}
	\sum_{j=1}^{\hat{m}}\mathbf{1}(|\{p\in Q_i\mid h_{i,j}(p)=1\}|>0) \geq \Omega\left(\varepsilon^{-2}L d^4 \log\left(\frac{dLk}{\delta}\right)\cdot \frac{L}{\delta}\cdot \frac{\hat{q}_i}{T_i(o)}\right).
	\end{align*}
\end{lemma}

\begin{proof}
	Consider a fixed $i\in I$. $\forall j\in[\hat{m}],$ by union bound, we have
	\begin{align*}
	\Pr_{h_{i,j}}\left[\exists p\in Q_i,h_{i,j}(p)=1\right]\leq  \frac{|Q_i|}{10^4kLT_i(o)} < 1,
	\end{align*}
	where the inequality follows by Lemma~\ref{lem:crucial_points_in_each_level}.
	Thus $\lceil 10^4 kLT_i(o)/|Q_i|\rceil\leq 2\cdot 10^4 kLT_i(o) / |Q_i|$. Let $b = 10 \cdot \lceil 10^4 kLT_i(o)/|Q_i|\rceil.$ Let $r = \lfloor \hat{m} / b\rfloor \geq \hat{m}\cdot |Q_i|/(2\cdot 10^5 kLT_i(o))-1$. Since $i\in I,$ we have $|Q_i|/T_i(o)\geq \frac{1}{2}\gamma$. Since $\hat{m}\geq 10^9 k L/\gamma$, we have $r\geq \hat{m}/(4\cdot 10^5 kL)\cdot |Q_i|/T_i(o).$
	For $s\in[r]$, we can define a random variable $Y_s$,
	\begin{align*}
	Y_s = \sum_{j=(s-1)\cdot b + 1}^{s\cdot b} \sum_{p\in Q_i} h_{i,j}(p).
	\end{align*}
	We have $\E[Y_s] = b\cdot |Q_i|/(10^4 k L T_i(o)).$ Thus, $E[Y_s]\in[10,20].$ Since $Y_s$ is a sum of several (at least) pairwise independent unit random variables, $\Var[Y_s]\leq 20$.
	Thus, by Chebyshev's inequality, we have 
	\begin{align*}
	\Pr\left[|Y_s-\E[Y_s]|\geq 9\right]\leq 20/81\leq 0.25.
	\end{align*}
	Define $X_s$ be the random variable such that $X_s=\mathbf{1}(Y_s\geq 1).$ Then $\E\left[\sum_{s\in [r]} X_s\right]\geq 0.75 r$.
	By Chernoff bound, we know that 
	\begin{align*}
	\Pr\left[\sum_{s\in[r]}X_s\leq 0.5 r\right]\leq 2^{-r/20}\leq 0.01\delta/L,
	\end{align*}
	where the last inequality follows by $r\geq \hat{m}/(4\cdot 10^5 kL)\cdot \gamma/2\geq 20\log(100L/\delta)$.
	
	Since $\hat{m}$ is sufficiently large, i.e.,
	\begin{align*}
	\hat{m}\geq \Omega\left(k\varepsilon^{-2} L^3d^4\log\left(\frac{dLk}{\delta}\right)\cdot\frac{1}{\delta\gamma}\right)\geq \Omega\left(k\varepsilon^{-3} L^4d^7\log\left(\frac{dLk}{\delta}\right)\cdot\frac{1}{\delta}\right),
	\end{align*}
	we have
	\begin{align*}
	r\geq \hat{m}/(4\cdot 10^5 kL)\cdot |Q_i|/T_i(o) \geq  \Omega\left(\varepsilon^{-2}L d^4 \log\left(\frac{dLk}{\delta}\right)\cdot \frac{L}{\delta}\cdot \frac{\hat{q}_i}{T_i(o)}\right),
	\end{align*}
	where the last inequality follows by $\hat{q}_i=\Theta(|Q_i|)$.
	Notice that 
	\begin{align*}
	\Pr\left[\sum_{j=1}^{\hat{m}}\mathbf{1}(|\{p\in Q_i\mid h_{i,j}(p)=1\}|>0) \geq 0.5r \right] \geq \Pr\left[\sum_{s\in[r]} X_s\geq 0.5 r\right]\geq 1-0.01\delta/L.
	\end{align*}
	Thus, with probability at least $1-0.01\delta/L$,
	\begin{align*}
	\sum_{j=1}^{\hat{m}}\mathbf{1}(|\{p\in Q_i\mid h_{i,j}(p)=1\}|>0) \geq \Omega\left(\varepsilon^{-2}L d^4 \log\left(\frac{dLk}{\delta}\right)\cdot \frac{L}{\delta}\cdot \frac{\hat{q}_i}{T_i(o)}\right).
	\end{align*}
	By taking the union bound over $i\in I$, we complete the proof.
\end{proof}

\begin{lemma}[Sampling stage succeeds]\label{lem:sampling_stage_suc}
	Consider $o\geq \OPT/16$. Conditioning on $\hat{f}:\bigcup_{i=0}^L G_i\rightarrow \mathbb{R}_+$ (in Algorithm~\ref{alg:sensitivity_sampler}) is good (Lemma~\ref{lem:good_samples}), $\hat{q}_0,\hat{q}_1,\ldots,\hat{q}_L$ (in Algorithm~\ref{alg:sensitivity_sampler}) are good (Lemma~\ref{lem:good_qi}), and the number of center cells is at most $100kL$, if \textsc{Sampling}$(o,\varepsilon,\delta)$ does not output \textbf{FAIL} before line~\ref{sta:uniform_sample_stage} of Algorithm~\ref{alg:sensitivity_sampler}, then with probability at least $1-\delta/5$, it will not output \textbf{FAIL}.
\end{lemma}
\begin{proof}
	According to Lemma~\ref{lem:num_samples_requested_each_level} and Lemma~\ref{lem:bound_of_tprime}, with probability at least $1-\delta/10$, $\forall i\in I$, the sampling procedure will not request too many samples from level $i$.
	According to Lemma~\ref{lem:num_samples_in_each_level_is_enough}, with probability at least $1-\delta/10$, the number of samples needed for each level $i\in I$ is enough. Thus, with probability at least $1-\delta/5$, the algorithm will not output \textbf{FAIL}.
\end{proof}

\begin{lemma}[Samples can fit into the space]\label{lem:samples_fit_into_the_space}
	Suppose $o\geq \OPT/16$. Conditioning on $\hat{f}:\bigcup_{i=0}^L G_i\rightarrow \mathbb{R}_+$ (in Algorithm~\ref{alg:sensitivity_sampler}) is good (Lemma~\ref{lem:good_samples}), if the total number of center cells is at most $100kL$,
	with probability at least $1-\delta/5$, \textsc{Sampling}$(o,\varepsilon,\delta)$ (Algorithm~\ref{alg:sensitivity_sampler}) will not output \textbf{FAIL} in line~\ref{sta:store_samples} nor line~\ref{sta:guarantee_storing_correctly}.
\end{lemma}
\begin{proof}
	Consider $i\in\{0,1,\ldots,L\}$ 
	and a cell $C\in G_i$ which is marked as crucial by \textsc{PointsEstimation}$(o,\varepsilon,\delta/2)$.
	\begin{align*}
	\E\left[\sum_{j=1}^{\hat{m}}\sum_{p\in C\cap Q} h_{i,j}(p)\right] &= \hat{m}\cdot |C\cap Q|/ (10^4 kL T_i(o)) \\
	&\leq O(\varepsilon^{-3}L^3d^7\log(dLk/\delta)\delta^{-1}).
	\end{align*}
	By Theorem~\ref{thm:lambda_wise_concentration},
	\begin{align*}
	\Pr\left[\sum_{j=1}^{\hat{m}}\sum_{p\in C\cap Q} h_{i,j}(p)>\E\left[\sum_{j=1}^{\hat{m}}\sum_{p\in C\cap Q} h_{i,j}(p)\right] + \varepsilon^{-3}L^3d^7\log(dLk/\delta)\delta^{-1} \right]\leq (\delta/\Delta^d)^5.
	\end{align*}
	Thus, by taking union bound, with probability at last $1-\delta/10$, $\forall i\in\{0,1,\ldots,L\},$ any cell $C\in G_i$ which is marked as crucial, the total number of points sampled in $C$ is at most $\beta$. Thus, with probability at least $1-\delta/20$, \textsc{Sampling}$(o,\varepsilon,\delta)$ (Algorithm~\ref{alg:sensitivity_sampler}) will not output \textbf{FAIL} in line~\ref{sta:guarantee_storing_correctly}.
	
	Consider $i\in\{0,1\ldots,L\}$. Let us analyze the number of cells in $G_i$ which contain at least $1$ sample points.
	The number of points which are not in the center cell is at most
	\begin{align*}
	\frac{\OPT}{(g_i/(2d))^2}\leq 400 k T_i(o)\cdot \OPT/o\leq 6400 k T_i(o).
	\end{align*} 
	Thus, the expected number of sampled points in non-center cell is at most $O(k\varepsilon^{-3}L^3d^7\log(dLk/\delta)\cdot 1/\delta)$. Since the number of center cells is at most $100kL$, by Theorem~\ref{thm:lambda_wise_concentration}, with probability at least $1-\delta/(100L)$, the number of cells in $G_i$ which contain at least $1$ sample points is $O(k\varepsilon^{-3}L^3d^7\log(dLk/\delta)\cdot 1/\delta)$.
	By taking union bound over all $i\in\{0,1,\ldots,L\},$ with probability at least $1-\delta/20$, $\forall i\in\{0,1,\ldots,L\}$, the number of sampled cell in $G_i$ is at most $\alpha$.
	
	Due to Lemma~\ref{lem:cell_point_storing}, with probability at least $1-\delta/10$, none of the \textsc{Storing}$(G_i,\alpha,\beta,0.1\delta/L)$ in line~\ref{sta:store_samples} of Algorithm~\ref{alg:sensitivity_sampler} will output \textbf{FAIL}.
	By union bound over all the failure probabilities, we complete the proof.
\end{proof}

\begin{lemma}[The overall success probability]\label{lem:the_overall_success_prob}
	Suppose $o\geq \OPT/16$ and the total number of center cells is at most $100kL$. With probability at least $1-4\delta/5$, \textsc{Sampling}$(o,\varepsilon,\delta)$ (Algorithm~\ref{alg:sensitivity_sampler}) will not output \textbf{FAIL}.
\end{lemma}
\begin{proof}
	Due to Lemma~\ref{lem:point_estimation_suc_prob}, \textsc{PointsEstimation}$(o,\varepsilon,\delta/2)$ in line~\ref{sta:run_dynamic_partition} will not output \textbf{FAIL} with probability at least $1-3\delta/20$. By Lemma~\ref{lem:good_estimation_for_each_cell} and Lemma~\ref{lem:good_qi}, with probability at least $1-\delta/5$, $\forall i\in\{0,1,\ldots,L\},$ either $\hat{q}_i\in |Q_i|\pm 0.1\varepsilon\gamma T_i(o)$ or $\hat{q}_i\in (1\pm 0.01\varepsilon)|Q_i|,$ and $\forall C\in G_i$, either $\hat{f}(C)\in|C\cap Q|\pm 0.1 T_i(o)$ or $\hat{f}(C)\in (1\pm 0.01) |C\cap Q|.$ Then by Lemma~\ref{lem:samples_fit_into_the_space}, with probability at least $1-\delta/5$, \textsc{Sampling}$(o,\varepsilon,\delta)$ will not output \textbf{FAIL} in line~\ref{sta:store_samples} nor line~\ref{sta:guarantee_storing_correctly}.
	Finally, according to Lemma~\ref{lem:sampling_stage_suc}, with probability at least $1-\delta/5$, the algorithm does not output \textbf{FAIL}. By taking union bound over all the bad events, with probability at least $1-3\delta/20-\delta/5-\delta/5-\delta/5\geq 1-4\delta/5$, \textsc{Sampling}$(o,\varepsilon,\delta)$ (Algorithm~\ref{alg:sensitivity_sampler}) will not output \textbf{FAIL}.
\end{proof}

\begin{lemma}[Total space of Algorithm~\ref{alg:sensitivity_sampler}]\label{lem:total_space_sampling_stream}
	\textsc{Sampling}$(o,\varepsilon,\delta)$ uses space at most $O(k\varepsilon^{-6}L^8d^{15}\delta^{-2}\cdot\log^4(kLd/(\varepsilon\delta)))$ bits.
\end{lemma}
\begin{proof}
	According to Lemma~\ref{lem:space_points_estimation}, \textsc{PointsEstimation}$(o,\varepsilon,\delta/2)$ in line~\ref{sta:run_dynamic_partition} of Algorithm~\ref{alg:sensitivity_sampler} takes the space $O(k\varepsilon^{-3}L^4d^5\cdot \log(kLd/(\varepsilon\delta)))$ bits. 
	According to Lemma~\ref{lem:cell_point_storing}, line~\ref{sta:store_samples} takes the total space:
	\begin{align*}
	O(L\cdot \alpha\beta dL\cdot\log^2(\alpha\beta/\delta)) = 
	O(k\varepsilon^{-6}L^8d^{15}\delta^{-2}\cdot\log^4(kLd/(\varepsilon\delta))).
	\end{align*}
\end{proof}

\subsection{The Final Algorithm}\label{sec:final}

Finally, we prove the guarantees of Algorithm~\ref{alg:final}, the final algorithm.

\begin{lemma}[Correctness and success probability]\label{lem:final_correctness}
	With probability at least $0.9$, \textsc{DynamicCoreset}$(\varepsilon)$ (Algorithm~\ref{alg:final}) outputs an $\varepsilon$-coreset for $Q$ and the size of the coreset is at most $O(k\varepsilon^{-2}L^2d^4\log(kLd))$.
\end{lemma}
\begin{proof}
	If $|Q|\leq 10000k$, then according to Lemma~\ref{lem:k_set}, with probability at least $0.999$, the entire data set $Q$ will be returned by line~\ref{sta:recover_all} in Algorithm~\ref{alg:final}.
	
	Consider the case when $|Q|\geq 10000k$. 
	According to Lemma~\ref{lem:num_center_cell}, with probability at least $0.94$, the total number of center cells is upper bounded by $100kL$. Now, we condition on this happens.
	Since $|Q|\geq 10000k$, we know that $\OPT\geq 1000k$. There exists $u\in[2dL]$ such that $o_u\in[\OPT/16,\OPT]$. According to Lemma~\ref{lem:the_overall_success_prob}, with probability at least $0.999$, \textsc{Sampling}$(o_u,\varepsilon,0.001/(dL))$ will not output \textbf{FAIL}. By Lemma~\ref{lem:correctness_of_sampler}, with probability at least $0.999$, the set $S$ returned by \textsc{Sampling}$(o_u,\varepsilon,0.001/(dL))$ is an $\varepsilon$-coreset and $|S|\leq O(k\varepsilon^{-2}L^2d^4)$.
	Thus, with probability at least $0.998$, \textsc{DynamicCoreset}$(\varepsilon)$ (Algorithm~\ref{alg:final}) will not output \textbf{FAIL}.
	Consider another $u'<u$. If \textsc{Sampling}$(o_{u'},\varepsilon,0.001/(dL))$ does not output \textbf{FAIL}, and the set $S'$ returned has size at most $h$, then according to Lemma~\ref{lem:correctness_of_sampler}, with probability at least $1-0.001/(dL),$ $S'$ is an $\varepsilon$-coreset for $Q$. By taking union bound over all the such $u'$, then with probability at least $0.999$, $S^*$ returned by \textsc{Sampling}$(o_{u^*},\varepsilon,0.001/(dL))$ is an $\varepsilon$-coreset for $Q$. By taking union bound over all the bad events, we complete the proof.
\end{proof}

\begin{lemma}[Total space needed for Algortihm~\ref{alg:final}]\label{lem:final_space}
	\textsc{DynamicCoreset}$(\varepsilon)$ (Algorithm~\ref{alg:final}) uses space at most $O(k\varepsilon^{-6}L^{11}d^{18}\log^4(kLd/\varepsilon))$ bits.
\end{lemma}
\begin{proof}
	\textsc{DynamicCoreset}$(\varepsilon)$ (Algorithm~\ref{alg:final}) runs $\Theta(dL)$ copies of \textsc{Sampling}$(o_u,\varepsilon,0.001/(dL))$. 
	By Lemma~\ref{lem:total_space_sampling_stream}, the total space needed is 
	\begin{align*}
	O(dL\cdot k\varepsilon^{-6}L^8d^{15}\cdot (dL)^2\cdot\log^4(kLd/\varepsilon)) = O(k\varepsilon^{-6}L^{11}d^{18}\log^4(kLd/\varepsilon))
	\end{align*}
	bits.
\end{proof}

\begin{proofof}{Theorem~\ref{thm:main}}
	The algorithm is shown by Algorithm~\ref{alg:final}. Lemma~\ref{lem:final_correctness} shows the correctness and the success probability of the algorithm. Lemma~\ref{lem:final_space} shows the total space needed by the algorithm.
\end{proofof}

\section{Why Do Previous Techniques Fail?}
\label{sec:previous algorithm fail}


In this section, we describe some previous techniques in more detail and explain why they fail in our setting.

\paragraph{Uniform sampling method.}
\citep{fis05} is one of the early papers using sampling procedures to solve dynamic streaming geometric problems. They showed that it is possible to use point samples from a dynamic point set to solve several geometric problems, e.g., Euclidean Minimum Spanning Tree.
However, they only showed how to implement \emph{uniform} sampling by using counting distinct elements and subsampling procedure as subroutines. In our setting, we require different sampling probabilities for different points.
Although the bottom-level uniform sampling scheme of ours is similar to theirs, our overall sampling method is more complicated.


\paragraph{Estimating the cost.}
\cite{i04} used a critical observation to \emph{estimate} the cost of $k$-median, that is: let $Z$ be a set of centers and $P$ be the point set, then $\cost(Z, P) = \int_{0}^{\infty} |P - B(Z, r)| \mathrm{d} r$, where $B(Z, r)$ is the union of balls of radius $r$ centered at all points in $Z$.
Then this integration is approximated by a summation with logarithmic levels, i.e.,  $\int_{0}^{\infty} |P - B(Z, r)| \mathrm{d} r \approx \sum_{i=0}^{\infty}  \big|P - B(Z, r^{(i+1)}) \big| \cdot \left(r^{(i+1)}-r^{(i)}\right)$, where $r^{(i)}= O(\epsilon(1+\epsilon)^i)$. The critical part is to estimate $\big| P-B(Z, r^{(i)}) \big|$.
\cite{i04} constructed a counting data structure based on grids with side length $O\left(r^{(i)}\right)$.
Then every input point is snapped to a grid point.  To obtain sufficiently accurate estimates for all $\left|P-B(Z, r^{(i)}) \right|$, the data structure needs to query $|Z| / \epsilon^{O(d)}$ many grid points per $Z$. Such a data structure is implemented using pair-wise independent hash functions, and uses memory $|Z|/\epsilon^{O(d)}$.
Notice that this method only gives an \emph{estimate} of the cost, and does not construct a coreset. In order to obtain a $k$-median solution, an exhaustive search is needed.
Furthermore, this technique fails to extend to $k$-means which does not have an integration formula for the cost function.


\paragraph{Exponential size coreset.}
\cite{fs05} constructed an $\epsilon$-coreset of size $k\epsilon^{-O(d)}$ for $k$-means and $k$-median.
They also used the same grid structure as we use.
A cell is marked as ``heavy'' if the cell contains enough points such that moving all points in the cell to the center of this cell incurs too much error in the optimal cost of $k$-means/median. Since the side-lengths of cells decrease as level increases, the number of points required to have this effect becomes larger. Eventually, all cells are non-heavy after some level. As such we also have a tree of heavy cells.
The coreset is constructed by looking at each heavy cell and assigning each point in its non-heavy children cells to its center.
It turns out that if we want an $\epsilon$-coreset, the threshold of non-heavy cells is exponential in $d$, i.e.,  each non-heavy cell in level $i$ cannot contain more than $\wt{O} \left(\epsilon^{O(d)} \cdot \OPT / 2^i \right)$ points. This small threshold gives rise to $\wt{O} \left(1/ \epsilon^{O(d)}\right)$ many heavy cells.

\paragraph{Insertion-only streams.}
Many of the previous \emph{insertion-only} streaming coreset construction algorithms (e.g., \citep{fs12}) heavily depend on a ``merge-reduce'' technique, i.e. reading some points in the stream, constructing a coreset, reading another part, constructing a new coreset, and then merging the two coresets. This procedure is repeated until the stream ends. This technique works well in the insertion-only streaming model, but it fails immediately when deletions are allowed. Although~\cite{bfl16} gave a new framework other than merge-reduce, their algorithm relies on a non-deleting structure of data streams as well.

\begin{figure}[t]
	\centering
	\includegraphics[width=0.8\textwidth]{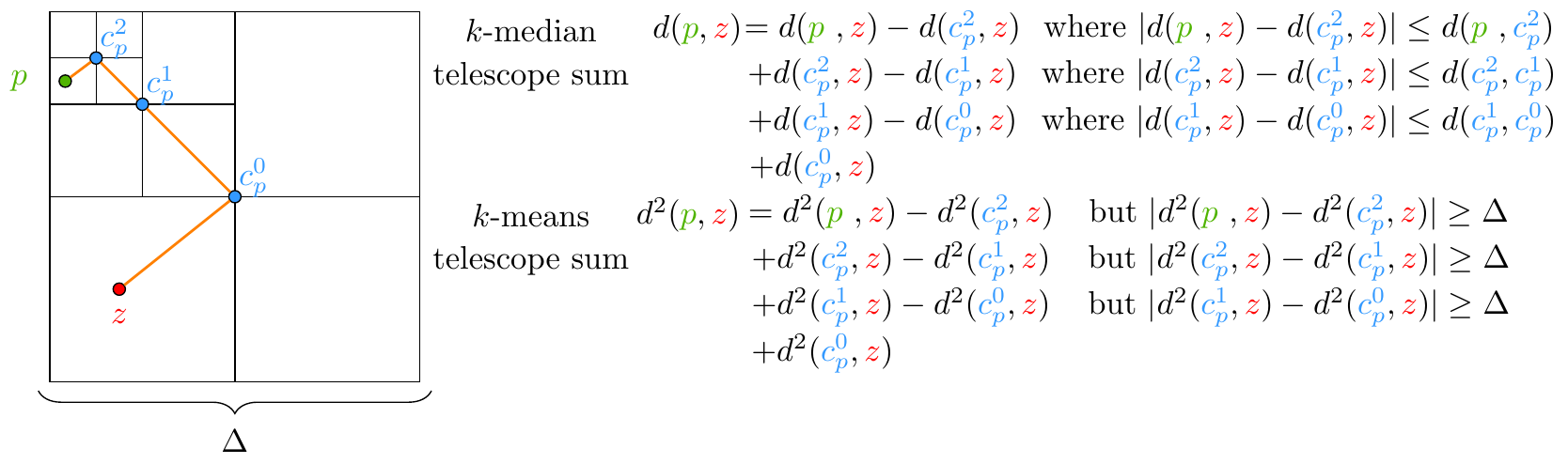}
	\caption{\small Telescope sum \citep{bflsy17} fails for $k$-means. In the $k$-median problem, for a fixed set of centers $Z$, the total cost can be written as a telescope sum $\sum_{p\in P} (\dist(c^i_p, Z) - \dist(c^{i-1}_p, Z))$. For each piece, $|\dist(c^i_p, Z) - \dist(c^{i-1}_p, Z)|$ is always upper bounded by $\dist(c^{i-1}_p, c^{i}_p)$ which is independent from the choice of $Z.$ However, in the $k$-means problem, the telescope sum of the total cost is $\sum_{p\in P} (\dist(c^i_p, Z)^2 - \dist(c^{i-1}_p, Z)^2)$. For each piece, the upper bound of $|\dist(c^i_p, Z)^2 - \dist(c^{i-1}_p, Z)^2|$ may depend on the location of $Z$, and it can be larger than $\Delta$ in the worst case.}
	\label{fig:telescope-sum}
\end{figure}

\paragraph{Algorithm for $k$-median only.}
Although some $k$-median coreset construction techniques can be easily extended to $k$-means (see e.g.~\citep{fs12,bfl16}), those constructions can only be implemented in the insertion-only streaming model.
\cite{bflsy17} gave a $k$-median coreset construction in the dynamic streaming model, but their construction cannot be extended to $k$-means directly, as we explain now. Their $k$-median algorithm heavily relies on writing the cost of each point as a telescope sum. 
For example, we consider the $1$-median problem. let $z$ be a candidate center point and $p \in P$ be a point, then
$\dist(p, z) = \dist(p, z) - \dist(c^{L-1}_p, z) + \dist(c^{L-1}_p, z) - \dist(c^{L-2}_p, z) + \cdots - \dist(c^0_p, z)$,
where each $c^i_p$ is the center of the cell in the $i$-th level containing $p$.
Therefore, the total $1$-median cost $\sum_{p\in P} \dist(p,z)$ of point set $P$ on $z$ can be split into $L$ pieces,
i.e., $\sum_{p\in P} (\dist(c^i_p, z) - \dist(c^{i-1}_p, z))$ for each $i \in [ L ]$. \cite{bflsy17} estimated each of the $L$ pieces by sampling points, i.e., let $S^i$ be the samples in the $i$-th level, then the estimator of $\sum_{p\in P} (\dist(c^i_p, z) - \dist(c^{i-1}_p, z))$ is $\sum_{p\in S^i} (\dist(c^i_p, z) - \dist(c^{i-1}_p, z))/\zeta_p^i$ where $\zeta_p^i$ is the probability that point $p$ is sampled. A crucial observation is that we have $|\dist(c^i, z) - \dist(c^{i-1}, z)| \le \Delta/2^i$ -- the cell size in level $i$ which is independent of the location of $z$. Using this nice upper bound on $|\dist(c^i, z) - \dist(c^{i-1}, z)|$, \cite{bflsy17} applied Bernstein inequality to get high concentration of the estimator $\sum_{p\in S^i} (\dist(c^i_p, z) - \dist(c^{i-1}_p, z))/\zeta_p^i$ with only $\tilde{O}(1/\epsilon^2)$ samples per level.
However, this framework does not work for $1$-means even though one can still write the telescope sum structure $\sum_p( \dist^2(c^i_p, z) - \dist^2(c^{i-1}_p, z))$ and can still setup an estimator $\sum_{p\in S^i}( \dist^2(c^i_p, z) - \dist^2(c^{i-1}_p, z))/\zeta_p^i$.
But $|\dist^2(c_p^i, z) - \dist^2(c_p^{i-1}, z)|$ is not upper bounded by the cell size anymore. Instead, it depends on the location of $z$. For example, it can be as large as $|\dist^2(c_p^i, z) - \dist^2(c_p^{i-1}, z)| \ge \Delta$. See Figure~\ref{fig:telescope-sum}.
If we apply Bernstein inequality here, 
we will need too many samples to save any space.


%





\end{document}